\documentclass[sigconf,nonacm]{acmart}

\usepackage{algorithmic}
\usepackage{graphicx}
\usepackage{xcolor}
\usepackage{color}
\usepackage{microtype}
\usepackage{booktabs}
\usepackage{enumitem}
\usepackage[font=footnotesize,labelfont=it]{subcaption}
\usepackage{xfrac}
\usepackage{multirow}
\usepackage{bbm}
\usepackage{chngcntr}
\usepackage{pdfrender}
\usepackage{mdframed}
\usepackage{geometry}
\usepackage{marginnote}

\setlist[itemize]{leftmargin=*,nosep}
\setlist[enumerate]{leftmargin=*,nosep}
\setlist[description]{leftmargin=1em,nosep}
\setlist{nolistsep}
\usepackage{pifont}
\newcommand{\xmark}{\ding{55}}

\newtheorem{lemma}{Lemma}

\newcommand{\unsim}{\mathord{\sim}}

\usepackage{listings}
\lstnewenvironment{Python}[1][]{\lstset{
language=Python,
basicstyle=\tiny\sffamily,
numberstyle=\color{gray},
commentstyle=\color[HTML]{228B22}\sffamily,
frame=single,
showstringspaces=false,
breaklines=true,
otherkeywords={True, False, None},
keywordstyle=\bfseries,
captionpos=t,
belowskip=0pt,
aboveskip=0pt,
alsoletter={.},
stringstyle=\color{red!50!black},
emphstyle={\color{violet!50!black}},
emphstyle={[2]\color{blue!50!black}},
#1
}}{}

\definecolor{revisionteal}{RGB}{135,174,174}

\mdfdefinestyle{revision}{backgroundcolor=revisionteal,hidealllines,linewidth=0,innerlinewidth=0,middlelinewidth=0,outerlinewidth=0,leftmargin=0,rightmargin=0,innerleftmargin=5pt,innerrightmargin=5pt,innertopmargin=2pt,innerbottommargin=0}

\definecolor{revisionrone}{RGB}{88,109,147}

\definecolor{revisionrtwo}{RGB}{103,148,106}

\definecolor{revisionrthree}{RGB}{124,115,174}

\definecolor{revisionrfour}{RGB}{181,85,85}

\newcommand{\redx}{\textcolor{red}{\xmark{}}}
\definecolor{nicergreen}{RGB}{57,151,92}
\newcommand*{\boldcheckmark}{%
  \textpdfrender{
    TextRenderingMode=FillStroke,
    LineWidth=.75pt, 
  }{\checkmark}%
}
\newcommand{\greencheck}{\textcolor{nicergreen}{\boldcheckmark{}}}

\AtBeginDocument{%
  \providecommand\BibTeX{{%
    \normalfont B\kern-0.5em{\scshape i\kern-0.25em b}\kern-0.8em\TeX}}}

\setcopyright{acmcopyright}
\copyrightyear{2021}
\acmYear{2021}
\acmDOI{10.1145/1122445.1122456}

\acmConference[SC'21]{Supercomputing '21: The International Conference for High Performance Computing, Networking, Storage, and Analysis}{November 14--19, 2021}{St. Louis, MO}
\acmBooktitle{Supercomputing '21: The International Conference for High Performance Computing, Networking, Storage, and Analysis,
  November 14--19, 2021, St. Louis, MO}
\acmPrice{15.00}
\acmISBN{978-1-4503-XXXX-X/18/06}

\begin{document}

\title{Clairvoyant Prefetching for Distributed Machine Learning I/O}

%
%
%

\author{Nikoli Dryden, Roman B\"{o}hringer, Tal Ben-Nun, Torsten Hoefler}
\thanks{Correspondence to: \href{mailto:ndryden@ethz.ch}{ndryden@ethz.ch}}
\affiliation{%
  \department{Department of Computer Science, ETH Z\"{u}rich, Switzerland}
}


\begin{abstract}
  I/O is emerging as a major bottleneck for machine learning training, especially in distributed environments.
Indeed, at large scale, I/O takes as much as 85\% of training time.
Addressing this I/O bottleneck necessitates careful optimization, as optimal data ingestion pipelines differ between systems, and require a delicate balance between access to local storage, external filesystems, and remote nodes.
We introduce NoPFS, a machine learning I/O middleware, which provides a scalable, flexible, and easy-to-use solution to the I/O bottleneck.
NoPFS uses \emph{clairvoyance}: Given the seed generating the random access pattern for training with SGD, it can exactly predict when and where a sample will be accessed.
We combine this with an analysis of access patterns and a performance model to provide distributed caching policies that adapt to different datasets and storage hierarchies.
NoPFS reduces I/O times and improves end-to-end training by up to $5.4\times$ on the ImageNet-1k, ImageNet-22k, and CosmoFlow datasets.

\end{abstract}

%



\maketitle

\vspace{-1em}
\section{Introduction}
\label{sec:intro}
\begin{table}[t]\footnotesize\fontsize{6.75}{\baselineskip}\selectfont
  \centering
  \setlength{\tabcolsep}{1.5pt}
  \begin{tabular}{lccccc}
    \toprule
    \multirow{2}{*}{Approach} & System & Dataset & Full & Hardware & Ease \\
    & scalability & scalability & randomization & independence & of use \\
    \midrule
    Double-buffering & \multirow{2}{*}{\redx{}} & \multirow{2}{*}{\greencheck{}} & \multirow{2}{*}{\greencheck{}} & \multirow{2}{*}{\redx{}} & \multirow{2}{*}{\greencheck{}} \\
    (e.g., PyTorch~\cite{paszke2019pytorch}) & & & & & \\
    \texttt{tf.data}~\cite{tensorflow2015-whitepaper,murray2021tf} & \redx{} & \greencheck{} & \redx{} & \redx{} & \greencheck{} \\
    Data sharding (e.g.,~\cite{kurth2018exascale}) & \greencheck{} & \redx{} & \redx{} & \redx{} & \greencheck{} \\
    DeepIO~\cite{zhu2018entropy} & \greencheck{} & \redx{} & \redx{} & \redx{} & \greencheck{} \\
    LBANN data store~\cite{jacobs2019parallelizing,oyama2020case} & \greencheck{} & \redx{} & \greencheck{} & \redx{} & \redx{} \\
    Locality-aware loading~\cite{yang2019accelerating} & \greencheck{} & \greencheck{} & \greencheck{} & \redx{} & \redx{} \\
    \textbf{NoPFS (this paper)} & \greencheck{} & \greencheck{} & \greencheck{} & \greencheck{} & \greencheck{} \\
    \bottomrule
  \end{tabular}
  \caption{Comparison of I/O frameworks.}
  \label{tab:io-frameworks}\vspace{-3em}
\end{table}

Training deep neural networks (DNNs) is an increasingly important workload on supercomputers, as deep learning (DL) is adopted by more fields.
Given the high cost of training, it is critical that every aspect be as efficient as possible~\cite{openai2018aiandcompute,strubell2019energy}.
Extensive work has been done to optimize training~\cite{bennun2019demystifying}, including dedicated hardware~\cite{jouppi2017datacenter,markidis2018nvidia}, compilers~\cite{xla,chen2018tvm}, optimizing operator primitives~\cite{chetlur2014cudnn,ivanov2021data}, and communication infrastructure~\cite{nccl,dryden2018aluminum,sergeev2018horovod,awan2019optimized,awan2017s}.

From the perspective of a DL framework, training a DNN involves three aspects: computation to execute the DNN; communication, to synchronize updates across nodes; and I/O, which provides the data and labels for training to each node.
The vast majority of work on optimizing training has focused on computation and communication.
Consequently, the performance bottleneck in training is shifting to I/O~\cite{pumma2019scalable,murray2021tf}.
Indeed, we find that when training ResNet-50~\cite{he2016deep} on ImageNet~\cite{deng2009imagenet} at scale, up to 85\% of runtime is I/O overhead, and we observe similar trends in other datasets.
As trends in compute capability continue with improving machine learning accelerators and datasets reach hundreds of millions~\cite{sun2017revisiting} to billions~\cite{mahajan2018exploring} of samples and terabytes~\cite{mathuriya2018cosmoflow,abu2016youtube,oyama2020case} to petabytes~\cite{abodo2018detecting} in size, this I/O bottleneck will only be exacerbated.

It is challenging to optimize training I/O, as stochastic gradient descent (SGD) randomly accesses (typically small) data samples.
This problem is especially acute for distributed training, where shared filesystem contention can be detrimental to performance.
Existing frameworks often overlap I/O with computation to reduce its overhead, but this is no longer sufficient.
Beyond this, ad hoc solutions such as limited lookahead and double-buffering~\cite{tensorflow2015-whitepaper,paszke2019pytorch,chien2018characterizing}, data sharding~\cite{goyal2017accurate,kurth2018exascale}, prestaging and in-memory caching~\cite{jacobs2019parallelizing,oyama2020case}, or modified access patterns~\cite{yang2019accelerating,zhu2018entropy} are used.
These have significant limitations, including poor scalability, requiring extra hardware, neglecting parts of the storage hierarchy, or deviating from full-dataset randomization.
All of these approaches can fail to fully utilize a machine's I/O subsystem (Tab.~\ref{tab:io-frameworks}, Sec.~\ref{sec:background}).

To address the I/O bottleneck, we introduce a new I/O middleware framework, the Near-optimal PreFetching System, NoPFS.\@
The key idea behind NoPFS is to use \emph{clairvoyance}~\cite{belady1966study}: Given the seed for the pseudorandom number generator (PRNG) that generates an access stream, we know exactly which process will access a given sample when, arbitrarily far in the future.
NoPFS analyzes the access stream to perform integrated prefetching and caching, rather than always reading from storage (Sec.~\ref{sec:access-patterns}).
It combines this with a performance model-driven distributed caching policy that uses both on-node storage hierarchies (e.g., RAM, node-local SSDs) and distributed memory (Secs.~\ref{sec:perfmodel},~\ref{sec:nopfs}).
This results in much-improved I/O performance, and overall improvements in runtime for ResNet-50~\cite{he2016deep} on ImageNet~\cite{deng2009imagenet} of up to $5.4\times$, up to $2.4\times$ on the larger ImageNet-22k dataset, and $2.1\times$ on CosmoFlow~\cite{mathuriya2018cosmoflow} (Sec.~\ref{sec:eval}).

Using NoPFS requires no changes to deep learning frameworks and changes to only a few lines of code in existing training scripts, as it presents an iterator-style interface to accessing data like standard data loaders (Fig.~\ref{fig:py-code}).
It can also automatically adapt to different datasets and machines, being applicable both to small-scale research clusters and large supercomputers.
Further, I/O subsystems are growing increasingly complex and differ between systems (Fig.~\ref{fig:hwinf}), a trend set to continue with future systems such as DAOS~\cite{daos} and Rabbit~\cite{rabbit}, making generic, performance model-driven systems even more attractive.

We additionally develop an I/O performance simulator to compare different I/O strategies in a variety of scenarios (Sec.~\ref{sec:perfsim}).
Beyond evaluating performance, this simulator can also be used to help design future systems for training workloads by analyzing which components (e.g., larger SSDs) have the largest impact on runtime.

When using NoPFS, I/O resources are fully utilized, alleviating the I/O bottleneck such that training is  necessarily limited by the dataset and hardware. Our key contributions are:
\begin{itemize}
\item We identify clairvoyance as a key insight for optimizing I/O and use this to perform a probabilistic analysis of access patterns and produce a near-optimal mapping of data to cache hierarchies.
\item We develop a performance model-driven distributed caching policy and implement it in NoPFS, an easy-to-use I/O middleware to optimize training I/O.
\item We significantly reduce I/O overhead, improving overall training time on ImageNet, ImageNet-22k, and CosmoFlow by up to $5.4\times$.
\end{itemize}

\begin{figure}
  \centering
  \includegraphics[width=\linewidth]{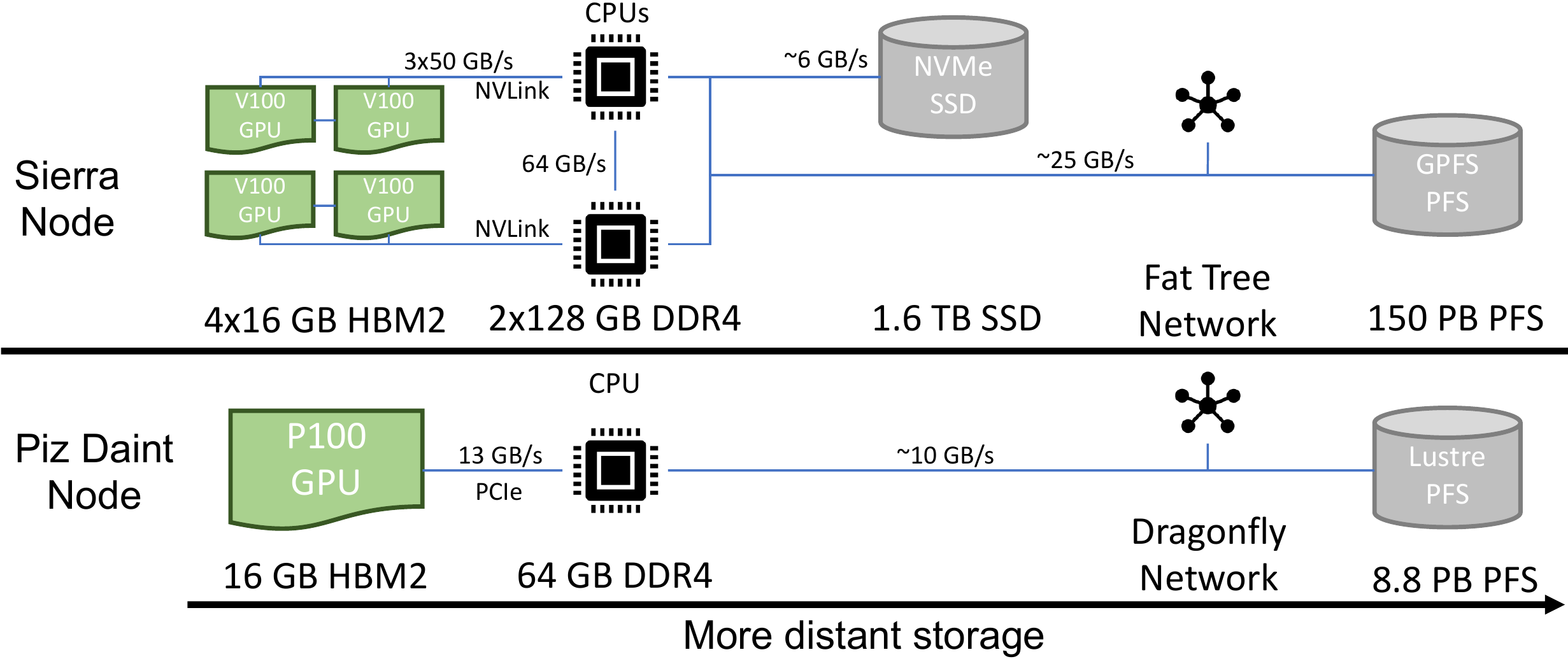}
  \caption{Storage hierarchies in modern supercomputers.}
  \label{fig:hwinf}
\end{figure}


\section{Background}
\label{sec:background}
Deep neural networks are almost always trained with variants of mini-batch SGD~\cite{bottou2018optimization}.
Training consists of many epochs; each epoch is a complete pass over the training dataset in a different, random order.
The samples that make up a given mini-batch are randomly selected without replacement from the entire training dataset.
This is typically implemented by assigning an index to each sample, randomly shuffling the indices each epoch, and then partitioning the shuffled indices into mini-batches.
Hence, a given sample is accessed exactly once in each epoch.
Given the seed used to shuffle the indices, we can exactly replicate the result of the shuffles, no matter the shuffle algorithm, and hence predict the access pattern, giving us \emph{clairvoyance}.
These access patterns hold across nearly all neural networks trained with mini-batch SGD.
For distributed training, we assume a data-parallel regime~\cite{bennun2019demystifying}, where a mini-batch is partitioned among workers.

\subsection{Data Sharding}
\label{subsec:data-sharding}

A relaxation of this regime sometimes used in practice is to assign each worker a shard (subset of the dataset) of data that fits in local storage, which may share samples with other workers (e.g.,~\citet{kurth2018exascale}).
This approach is typically adopted in order to mitigate the I/O overheads of shared storage by only using local storage.
However, this has three major limitations:
(1) The dataset must fit in the aggregate local storage. If it does not, then samples or even entire rare classes may be missed, impacting learning.
Further, on many systems, local storage is small; e.g., Piz Daint~\cite{daint} has 64 GB/node and Fugaku~\cite{fugaku} only 32 GB/node, which must also hold the model and activations when training.
(2) This can change the randomization performed by SGD, which may impact learning. It has been observed that full-dataset randomization and without-replacement sampling performs better~\cite{goyal2017accurate,gurbuzbalaban2019random}.
(3) This does not fully utilize the storage hierarchy, as it neglects distributed memory.
For example, while accessing a node-local SSD may be faster than a contended PFS, it may be faster still to access samples from a remote node's memory, as modern network bandwidth (often 10+ GB/s) is higher than SSD read bandwidth (2--10 GB/s); an ideal solution would use both.
In this work, we will focus on full-dataset randomization to avoid any issues with learning.

\subsection{Machine Learning I/O Frameworks}
\label{subsec:io-frameworks}

\begin{figure}[t]
  \centering
  \includegraphics[width=\linewidth]{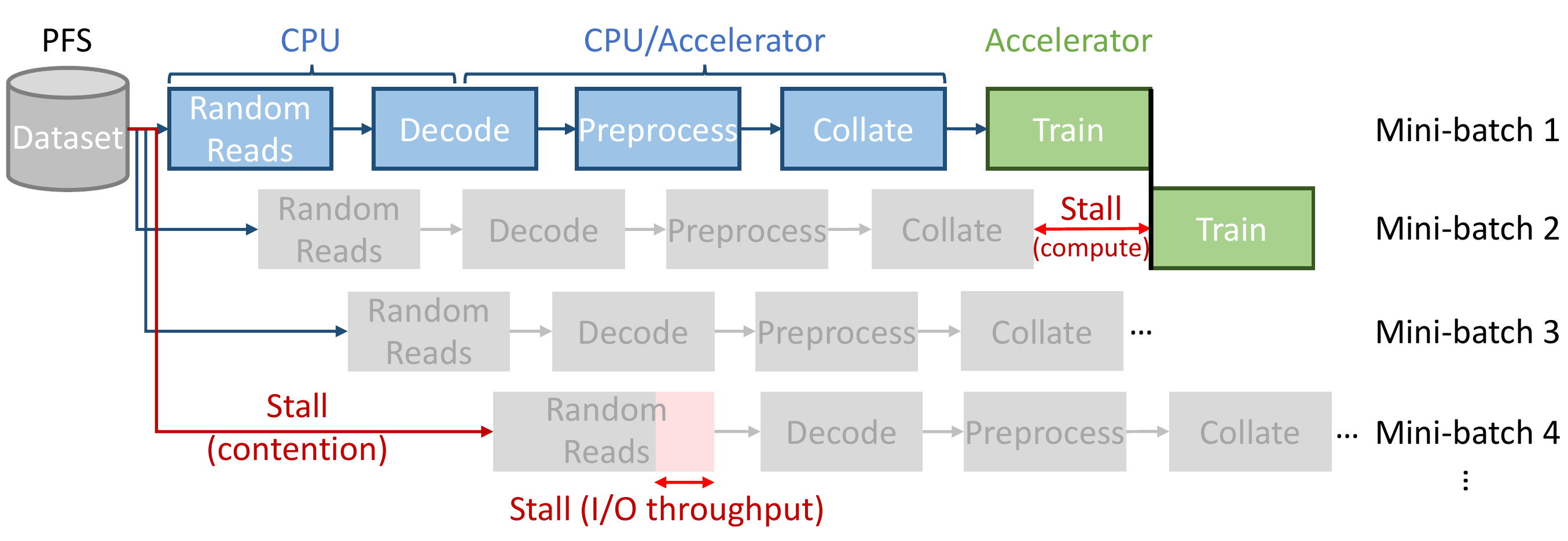}
  \caption{Overview of I/O pipelines and potential issues.}
  \label{fig:io-pipeline}
\end{figure}

I/O for training deep neural networks is a complex task and typically consists of multiple stages.
At the highest level, ``I/O'' involves reading samples from storage, preprocessing them, where preprocessing may entail multiple stages itself, such as decoding images, tokenizing text, normalization, or data augmentation, and finally collating them into a mini-batch for training (see Fig.~\ref{fig:io-pipeline}).
Stalls at any point in this process can impact training time.

There are a number of solutions for optimizing the preprocessing stage, such as memory mapping optimizations~\cite{pumma2019scalable}.
In general, these are orthogonal to optimizations for the read stage.
DALI~\cite{dali} includes both preprocessing optimizations and a file cache.
We focus primarily on optimizing this first stage, which we will refer to as ``I/O''.
In this setting, we will assume that the dataset begins in a shared storage location such as a parallel filesystem (PFS), and that training workers run on compute nodes that all have access to the storage.
This matches the MLPerf-HPC requirements~\cite{mlperf-hpc}.

We identify several key characteristics of I/O frameworks:
\begin{description}
\item[System scalability] Whether additional resources can be productively used when scaling to many nodes.
\item[Dataset scalability] Whether arbitrarily large datasets (much larger than aggregate node storage) are supported.
\item[Full randomization] Whether sample selection is randomized over the entire dataset without replacement in each epoch.
\item[Hardware independence] Whether special hardware (e.g., node-local SSDs) is used if present, but is not required.
  Storage hierarchies are complex and often differ between systems (Fig.~\ref{fig:hwinf}), making this especially important.
\item[Ease of use] Whether significant effort is needed to incorporate the framework in workflows.
\end{description}

We summarize existing approaches, along with NoPFS, according to these characteristics in Tab.~\ref{tab:io-frameworks}.
All of these approaches are capable of double-buffering, where fetching the next mini-batch is overlapped with computation, and using multiple threads to fetch and preprocess samples.
This approach is taken by PyTorch's built-in \texttt{DataLoader}~\cite{paszke2019pytorch}.
TensorFlow's \texttt{tf.data}~\cite{murray2021tf} extends this with longer-range lookahead, but typically performs data shuffling only in a limited-size buffer instead of over the entire dataset.
These two approaches have poor system scalability, as workers contend for access to shared storage.
Data sharding is widely used in practice, generally with ad hoc data staging scripts, but is necessarily limited by available system storage.
None of the existing approaches are hardware independent; either they require additional hardware, such as SSDs, or they neglect the hardware when it is available.


\section{I/O Access Patterns}
\label{sec:access-patterns}
We first review the prior work on prefetching and caching algorithms, and then analyze I/O access patterns in training.
For caches, given the access stream $R$, the optimal schedule is given by B\'{e}l\'{a}dy's clairvoyant algorithm~\cite{belady1966study}, which replaces the data that will not be needed for the longest time.
However, it is much more challenging to derive an optimal schedule for integrated prefetching and caching~\cite{cao1995study}.
There exist efficient algorithms for finding the optimal schedule in the case of a single processor and disk, and approximation algorithms for the case of a single processor with a parallel disk system~\cite{kimbrel2000near,albers2000minimizing,albers2001minimizing,jin2002simple,albers2005integrated,cao1994application}.
Unfortunately, finding an optimal schedule for the parallel disk case is NP-hard~\cite{ambuhl2003parallel}.
Similar work has been done in the context of caches for multi-core processors, where there are results for cache hierarchies, although optimal algorithms are again NP-hard~\cite{kamali2021beyond,agrawal2021tight,kamali2020multicore,bender2014cache,katti2012competitive,lopez2012paging,hassidim2010cache}.
Our case, where there are multiple processors each possibly with multiple levels of cache, is even more challenging.

Nevertheless, we can adapt ideas from these algorithms to our situation.
It can be shown that any optimal prefetching and caching strategy for a single processor and disk must follow four rules~\cite{cao1995study}:
\begin{enumerate}
\item \textbf{Optimal prefetching:} Every prefetch should fetch the next sample in $R$ that is not in the cache.
\item \textbf{Optimal replacement:} Every prefetch should discard the sample whose next use is furthest in the future.
\item \textbf{Do no harm:} Never discard sample $A$ to prefetch sample $B$ when $A$ will be used before $B$.
\item \textbf{First opportunity:} Never prefetch-and-replace when the same operation could have been done previously.
\end{enumerate}
Some of these rules can be generalized to the case of multiple disks~\cite{kimbrel2000near}, or relaxed while still producing good approximations~\cite{jin2002simple}. NoPFS is able to implement Rule 1 exactly and approximates the remaining rules within a limited time horizon, using the fact that a sample is accessed exactly once per epoch.

\subsection{Distribution of Accesses}
\label{subsec:acc-dist}

\begin{table}[t]\scriptsize
  \centering
  \begin{tabular}{lll}
    \toprule
    \textbf{Variable} & \textbf{Unit} & \textbf{Definition} \\
    \midrule
    $R$ & & Access sequence of a worker \\
    $N$ & & Number of workers \\
    $E$ & & Number of epochs \\
    $F$ & & Number of samples in dataset \\
    $c$ & MB/s & Compute throughput \\
    $\beta$ & MB/s & Preprocessing rate \\
    $b_c$ & MB/s & Inter-worker network bandwidth \\
    $t(\gamma)$ & MB/s & Random agg. read throughput (with $\gamma$ clients) of the PFS \\
    $p_j$ & & Number of threads for prefetching to storage class $j$ \\
    $d_j$ & MB & Capacity of storage class $j$ \\
    $r_j(p)$ & MB/s & Random agg. read throughput of storage class $j$ ($p$ reader threads) \\
    $w_j(p)$ & MB/s & Random agg. write throughput of storage class $j$ ($p$ writer threads) \\
    $D$ & MB & Total local storage of a worker \\
    $s_k$ & MB & Size of sample $k$ \\
    $S$ & MB & Size of dataset \\
    $B$ & & Batch size \\
    $T$ & & Number of iterations per epoch \\
    $t_{i,f}$ & & Time elapsed when worker $i$ consumes sample $R_f$ \\
    \bottomrule
  \end{tabular}
  \caption{Notation used throughout paper.}
  \label{tab:definition-summary}
  \vspace{-2em}
\end{table}

NoPFS utilizes a second key observation about the access pattern: Although each sample is read once per epoch, the number of times the same worker will read that sample over $E$ epochs of training varies depending on the random seed.
Exploiting this access frequency disparity allows us to devise better cache policies.

To formalize this, consider a fixed worker and sample, and let $X_e$ be the probability that worker will access the sample in epoch $e$.
For $N$ workers and $E$ epochs, we have that $X_e \unsim \mathrm{Bernoulli}(\frac{1}{N})$ and the access frequency $X$ of this sample is $X = \sum_{e=1}^E X_e$.
As the $X_e$ are independent Bernoulli random variables with the same success probability, we have that $X \unsim \mathrm{Binomial}(E, \frac{1}{N})$.
Thus the mean of the distribution is $\mu = \mathbb{E}[X] = \frac{E}{N}$ and the probability that the access frequency is greater than $\mu$ by a factor $\delta$ (and hence the sample will be accessed more often by the worker) is
\[ P(X > (1 + \delta) \mu) = \sum_{k=\lceil (1+\delta) \mu \rceil}^K \binom{E}{k} {\left( \frac{1}{N} \right)}^k {\left( \frac{N-1}{N} \right)}^{K-k}. \]
However, we are primarily interested in the number of samples that will be accessed more often by a worker, which is the sum over all samples of $\mathbbm{1}_{X>(1+\delta)\mu}$.
Then, using that expectation is linear, the expected value is given by $F \cdot P(X > (1 + \delta) \mu)$, where $F$ is the size of the dataset.
We verified that this works well using Monte Carlo simulations.
As an example, consider a standard ImageNet-1k training run with $N = 16$, $E = 90$, $F = 1{,}281{,}167$, and $\delta = 0.8$.
Our estimate gives an expected value of $\unsim 31{,}635$: although each sample is read $6$ times on average by a worker, around $31{,}635$ samples will be accessed more than $10$ times.
The distribution from a Monte Carlo simulation is shown in Fig.~\ref{fig:acc-mc}.
The number of samples accessed more than $10$ times is $31{,}863$, closely agreeing with the calculations.

\begin{figure}
  \centering
  \includegraphics[width=0.9\linewidth]{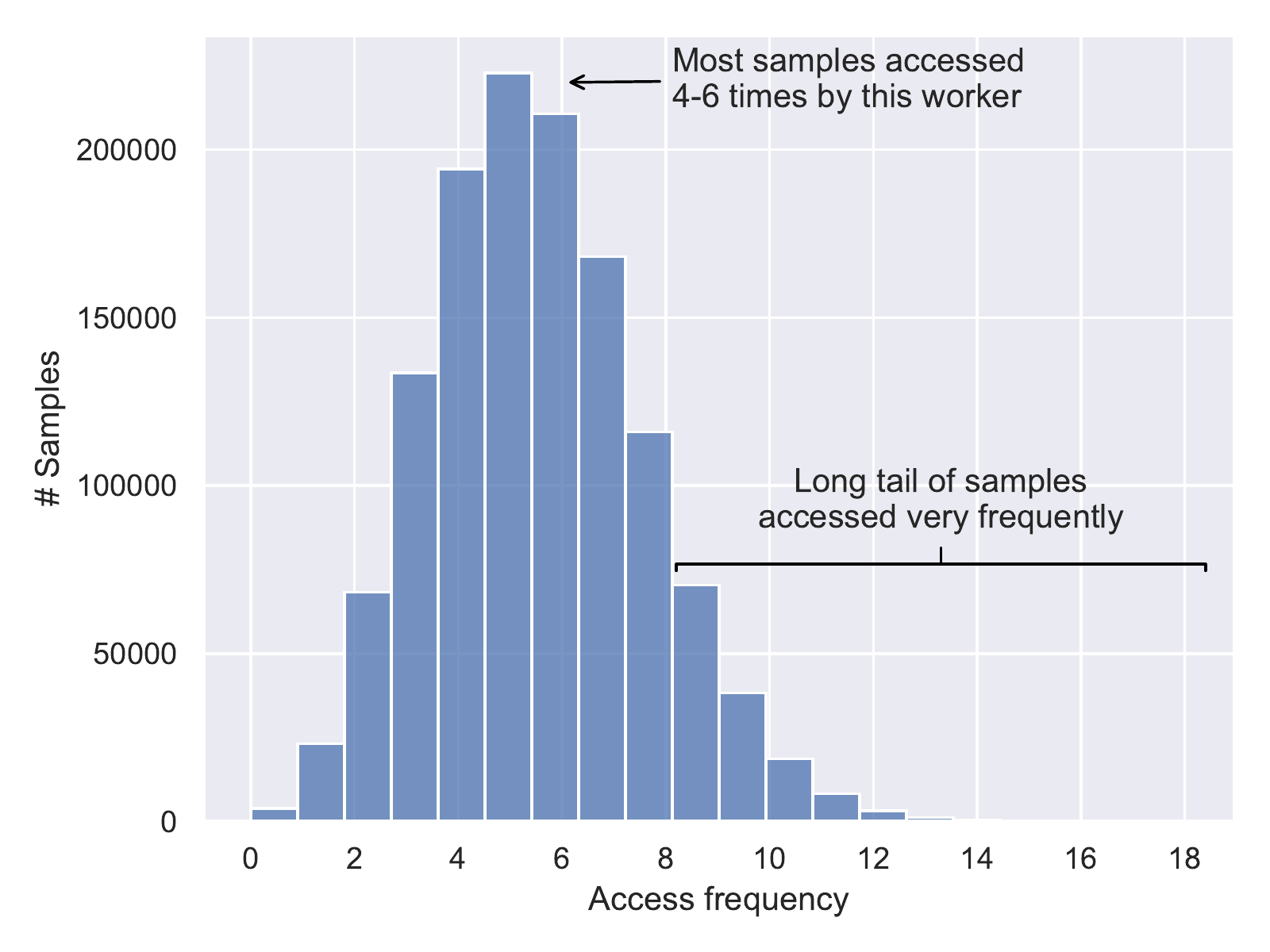}
  \caption{Simulation of access frequency for a single process (of 16) when training for 90 epochs on ImageNet-1k.}\vspace{-2em}
  \label{fig:acc-mc}
\end{figure}

As each sample is accessed exactly $E$ times by fully-randomized SGD without replacement, if some worker access a sample more (or less) frequently, then some other worker must access it less (or more) frequently.
We formalize this as follows:
\begin{lemma}
  \label{lemma:access-bounds}
  If a worker accesses a sample $\lceil (1 + \delta) \frac{E}{N} \rceil$ times (resp. $\lfloor (1 - \delta) \frac{E}{N} \rfloor$ times), at least one other worker will access the sample at most $\lceil (\frac{N - 1 - \delta}{N - 1}) \frac{E}{N} \rceil$ (resp.\ at least $\lfloor (\frac{N - 1 + \delta}{N - 1}) \frac{E}{N} \rfloor$) times.
\end{lemma}
\begin{proof}
  Assume towards a contradiction that every other worker accesses the sample $\lceil (\frac{N - 1 - \delta}{N - 1})\frac{E}{N} \rceil + 1$ times for some $E$, $\delta \in [0, N - 1]$, and $N > 1$.
  Then the total accesses to this sample are
  \begin{align*}
    \left\lceil (1 + \delta) + \frac{E}{N} \right\rceil + (n - 1) \left( \left\lceil \left( \frac{N - 1 - \delta}{N - 1} \right) \frac{E}{N} \right\rceil + 1 \right) &\geq \\
    (1 + \delta) \frac{E}{N} + (N - 1) \left( \left( \frac{N - 1 - \delta}{N - 1} \right)  \frac{E}{N} + 1 \right) &= \\
    (1 + \delta) \frac{E}{N} + E - (1 + \delta) \frac{E}{N} + N - 1 &= \\
    E + N - 1 &> E.
  \end{align*}
  This contradicts that every sample is accessed exactly $E$ times during training. The proof of the other bound is symmetric.
\end{proof}


\section{Performance Modeling}
\label{sec:perfmodel}
The I/O access frequency distribution allows us to identify frequently used samples to cache on a worker.
We now turn to our performance model of training I/O, which will allow us to decide where to cache samples and where to fetch them from to minimize training time.
These two analyses combined form the basis for NoPFS (Sec.~\ref{sec:nopfs}).
It also enables us to develop a simulator to compare I/O frameworks, identify bottlenecks, and help design future systems for training workloads (Sec.~\ref{sec:perfsim}).

First we introduce some notation to define the compute environment; all associated quantities can be measured with simple benchmarks such as training microbenchmarks, STREAM~\cite{mccalpin1995memory}, FIO~\cite{fio}, and IOR~\cite{ior}.
To simplify presentation, we will assume there is one worker per compute node (this is not necessary in practice).

Let there be $N$ workers, where each worker $i$ has:
\begin{itemize}
\item $c$ [MB/s]: Compute throughput for training.
  This depends on the details of the neural network, hardware, and software.
  We model $c$ as MB/s as it directly relates to I/O subsystem parameters; if it is known only in terms of samples/second, it can be approximated by multiplying this by the average file size.
  If samples are resized during preprocessing, the original size should be used.
\item $\beta$ [MB/s]: Data preprocessing rate.
\item We will assume there is no network congestion.
\item $b_c$ [MB/s]: Inter-worker network bandwidth.
\item $t(\gamma)$ [MB/s]: Random aggregate read throughput of the PFS, as a function of the number of readers $\gamma$. This depends on $\gamma$ as PFS bandwidth is heavily dependent on the number of clients~\cite{chowdhury2019characterization}.
\end{itemize}

To account for the storage diversity present in current and upcoming systems, we will assume there are $J$ distinct storage classes which group similar storage media.
E.g., a storage class can represent RAM, SSDs, HDDs, shared global burst buffers, or emerging NVRAM technologies.
Storage class $0$ is defined to be the staging buffer, a (usually small) in-memory buffer that is shared with the machine learning framework.
Storage class $j$ is characterized by:
\begin{itemize}
\item $d_j$ [MB]: Capacity of storage class $j$. The total local storage of a worker is therefore $D = \sum_{j=1}^{J} d_j$.
\item $r_j(p)$ and $w_j(p)$ [MB/s]: Random aggregate read and write throughput for storage class $j$ with $p$ threads.
\item $p_j$: Number of threads used for prefetching data to storage class $j$.
  We assume there is always at least one thread for prefetching to the staging buffer, i.e., $p_0 \geq 1$.
\end{itemize}
We model throughput as a function of $p$ as for many storage devices, a single thread cannot saturate its bandwidth.

Let our training dataset consist of $F$ samples, where sample $k$ has size $s_k$.
Each sample may have a different size.
The size of the whole dataset is $S = \sum_{k=1}^{F} s_k$.
We can have that $S > D$, where it is not possible for the dataset to be stored on a single worker, or $S > ND$, where the dataset cannot be stored across all workers.
The mini-batch size is $B$ and there are $E$ epochs.
One epoch consists of $T = \lfloor \frac{F}{B} \rfloor$ iterations (or $\lceil \frac{F}{B} \rceil$ if we keep the last, small iteration).

At iteration $h$, $1 \leq h \leq ET$, we process a batch $B^h \subseteq \{1, \ldots, F\}$ and worker $i$ processes its local batch $B^{h,i} \subseteq B^h$.
We write $b_i = |B^{h,i}|$.
As each sample is read exactly once in an epoch, the sets $B^k$ for $rT \leq k \leq (r+1)T$, for some $r \in \mathbb{N}$, are pairwise disjoint, which implies the same for the $B^{k,i}$.
For data-parallelism, we have that $B^{h,1}, \ldots, B^{h,N}$ partition $B^h$. (Adapting this to other training regimes, e.g., model-parallelism, is straightforward.)

Lastly, we write $B^{h,i}_\ell$ to be the $\ell$th sample in worker $i$'s $h$th batch.
Then the worker's access stream is $R = (B^{1,i}_1, B^{1,i}_2, \ldots, B^{1,i}_{b_i}, B^{2,i}_1, \ldots)$.

\begin{figure}
  \centering
  \includegraphics[width=0.9\linewidth]{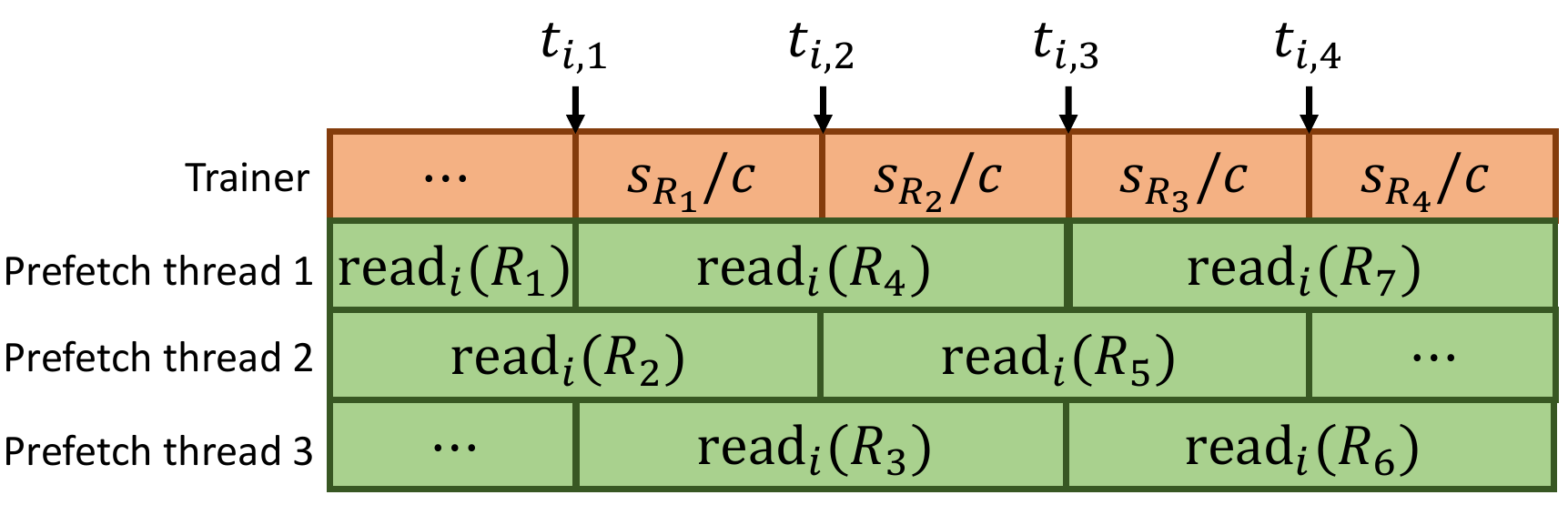}
  \caption{Performance model access times.}\vspace{-1.5em}
  \label{fig:perf-access}
\end{figure}

We now define the key metric of our model, $t_{i,f}$, the time elapsed when worker $i$ consumes $R_f$, the $f$th entry of $R$:
\[ t_{i,f} = \max \left( \mathrm{avail}_i(f), t_{i,f-1} + \frac{s_{R_{f-1}}}{c} \right), \]
where $\mathrm{avail}_i(f)$ is the time $R_f$ is available in the staging buffer.
This is illustrated in Fig.~\ref{fig:perf-access}.
Assuming threads are load balanced, we have
\[ \mathrm{avail}_i(f) = \frac{\sum_{k=1}^f \mathrm{read}_i(R_k)}{p_0}. \]
We define $\mathrm{read}_i(k) = \mathrm{fetch}_i(k) + \mathrm{write}_i(k)$ as the time to read the $k$th sample into the staging buffer.
Here, $\mathrm{fetch}_i(k)$ is the time to fetch the sample into memory and $\mathrm{write}_i(k)$ the time to preprocess and store it in the staging buffer.
$\mathrm{write}_i(k)$ does not depend on the data source, and is
\[ \mathrm{write}_i(k) = \max \left(  \frac{s_k}{\beta}, \frac{s_k}{w_0(p_0)/p_0} \right), \]
where we assume preprocessing and writing can be pipelined in parallel.
For fetching data, there are three cases, and we use the fastest applicable one:
\begin{enumerate}
\item Reading from the PFS, while $\gamma - 1$ other workers do as well: $\mathrm{fetch}_{i,0,0}(k) = s_k / (t(\gamma)/\gamma)$.
\item Reading from another worker's storage class $j$: $\mathrm{fetch}_{i,1,j}(k) = s_k / \min(b_c, r_j(p_j) / p_j)$.
\item Reading from its local storage class $j$ (assuming the sample is present): $\mathrm{fetch}_{i,2,j}(k) = s_k / (r_j(p_j) / p_j)$.
\end{enumerate}

This performance model drives runtime selection of data fetching and caching in NoPFS.
Note, in practice, because of remote data fetching, we may not know the exact number of threads accessing a local storage class.
However, this generally does not change the rank ordering due to disparities in speed of access.


\section{N{\lowercase{o}}PFS}
\label{sec:nopfs}
We now present the design and implementation of NoPFS, which combines the aforementioned performance model with our analysis of access patterns to provide distributed caching and prefetching.

\subsection{Design}
\label{subsec:nopfs-design}

\begin{figure}[t]
  \centering
  \includegraphics[width=0.9\linewidth]{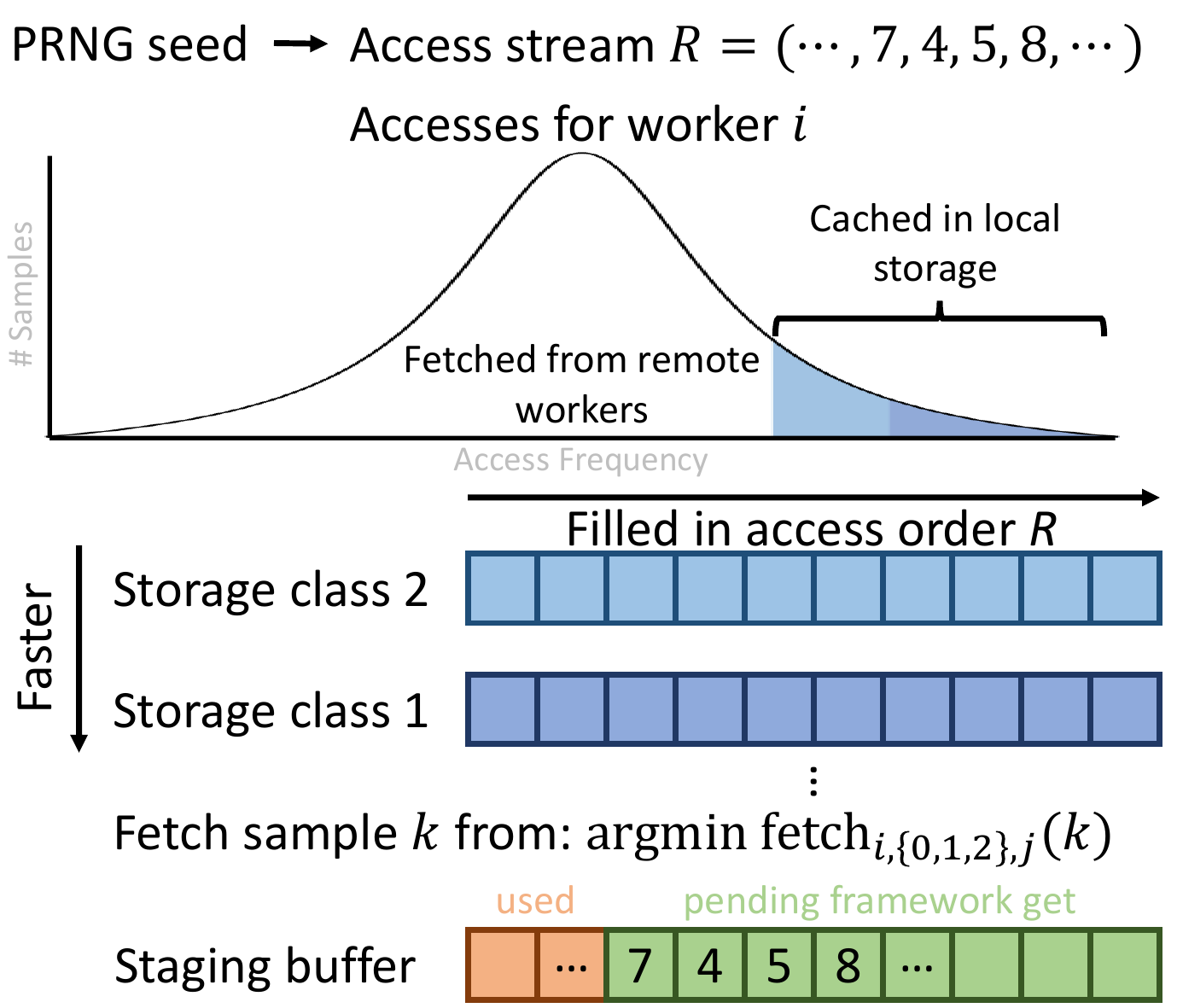}
  \caption{Overview of NoPFS's prefetching/caching policy.}
  \label{fig:nopfs-design}
\end{figure}

NoPFS needs to answer several questions to implement its prefetching and caching policy:
(1) Which samples should be fetched to the staging buffer when?
(2) Where should these samples be fetched from?
(3) Which samples should be assigned to which storage class, and what order should they be prefetched in?
We will discuss each of these in turn; because NoPFS uses clairvoyance and a performance model, the solutions are near-optimal.
The overall policy is summarized in Fig.~\ref{fig:nopfs-design}.

As we know the PRNG seed, we can exactly compute $R$, and with this prefetch data in the correct access order into the staging buffer (satisfying Rule 1).
Once a sample is read, a worker will access it again at the earliest in the next epoch, and every sample that follows in the current epoch is necessarily accessed earlier.
Therefore, we can approximate Rules 2--4 by immediately dropping samples from the staging buffer after access, freeing up space for samples that (with high probability) will be accessed sooner.

While this determines which samples to fetch to the staging buffer at what time, we need to use our performance model to decide from where to fetch samples.
Because we know $R$ for each worker, every worker knows where every sample is cached, and we can select the location to fetch from that requires minimal time.

Finally, we define the strategy used by other storage classes.
Suppose the worst case where a worker always waits before consuming a sample.
Then the total training time is
\[ t_{i,|R|} = \mathrm{avail}_i(|R|) = \frac{\sum_{k=1}^{|R|}(\mathrm{fetch}_i(R_k) + \mathrm{write}_i(R_k))}{p_0}. \]
We fill the other storage classes to minimize this.
If we ignore fixed terms in the strategy, we need to compute $\min \sum_{k=1}^{|R|} \mathrm{fetch}_i(R_k)$.
As we can select the fastest location to fetch from, this becomes
\[ \sum_{k=1}^{|R|} \frac{s_{R_k}}{\max(t(\gamma)/\gamma, \min(b_c, r_{j_r}(p_{j_r})/p_{j_r}), r_{j_\ell}(p_{j_\ell})/p_{j_\ell})}, \]
where $j_r$ and $j_\ell$ are the fastest remote and local storage class of sample $R_k$, respectively.
If a file is not available locally or at any remote worker, we define the respective throughput to be 0.
Letting $r_k$ be the access frequency of sample $k$ and assuming a static system (i.e., samples are already loaded in storage classes and no parameters change), this becomes a sum over all samples:
\[ \sum_{k=1}^F \frac{r_k s_k}{\max(t(\gamma)/\gamma, \min(b_c, r_{j_r}(p_{j_r})/p_{j_r}), r_{j_\ell}(p_{j_\ell})/p_{j_\ell})}. \]
Assuming that samples are similarly sized, we can conclude:
\begin{enumerate}
\item When $r_k$ is large (i.e., a worker accesses a sample frequently), we want $r_{j_\ell}(p_{j_\ell})$ to be large, and therefore should cache the sample in a fast local storage class.
\item As $t(\gamma)/\gamma$ is often constant or decreasing with many readers, we want to minimize $\gamma$ to reduce PFS contention.
  We also want $r_{j_r}(p_{j_r})$ to be large for samples where $r_{j_\ell}(p_{j_\ell})$ is small (i.e., samples not cached locally, or in slow storage).
  These imply samples should be well-distributed among workers.
\end{enumerate}

Recall that the access frequency $r_k$ varies for different $k$ and Lemma~\ref{lemma:access-bounds} implies that when $r_k$ is small on one worker, it will be large on at least one other worker (and vice versa).
We thus use $r_k$ to make the fetching decision: A worker fetches samples with the largest $r_k$ to its fastest storage class, and so on for slower classes until either it has cached the entire dataset or filled its local storage.

The last step is to define the fetch order.
Our analysis has thus far assumed all storage classes have already been filled, but this would require a potentially costly prestaging step that cannot be overlapped with training.
We follow Rule 1 and use $R$ to ensure we always fetch the samples in order of access.

\subsection{Implementation}
\label{subsec:nopfs-impl}

\begin{figure}[t]
  \centering
  \includegraphics[width=0.8\linewidth]{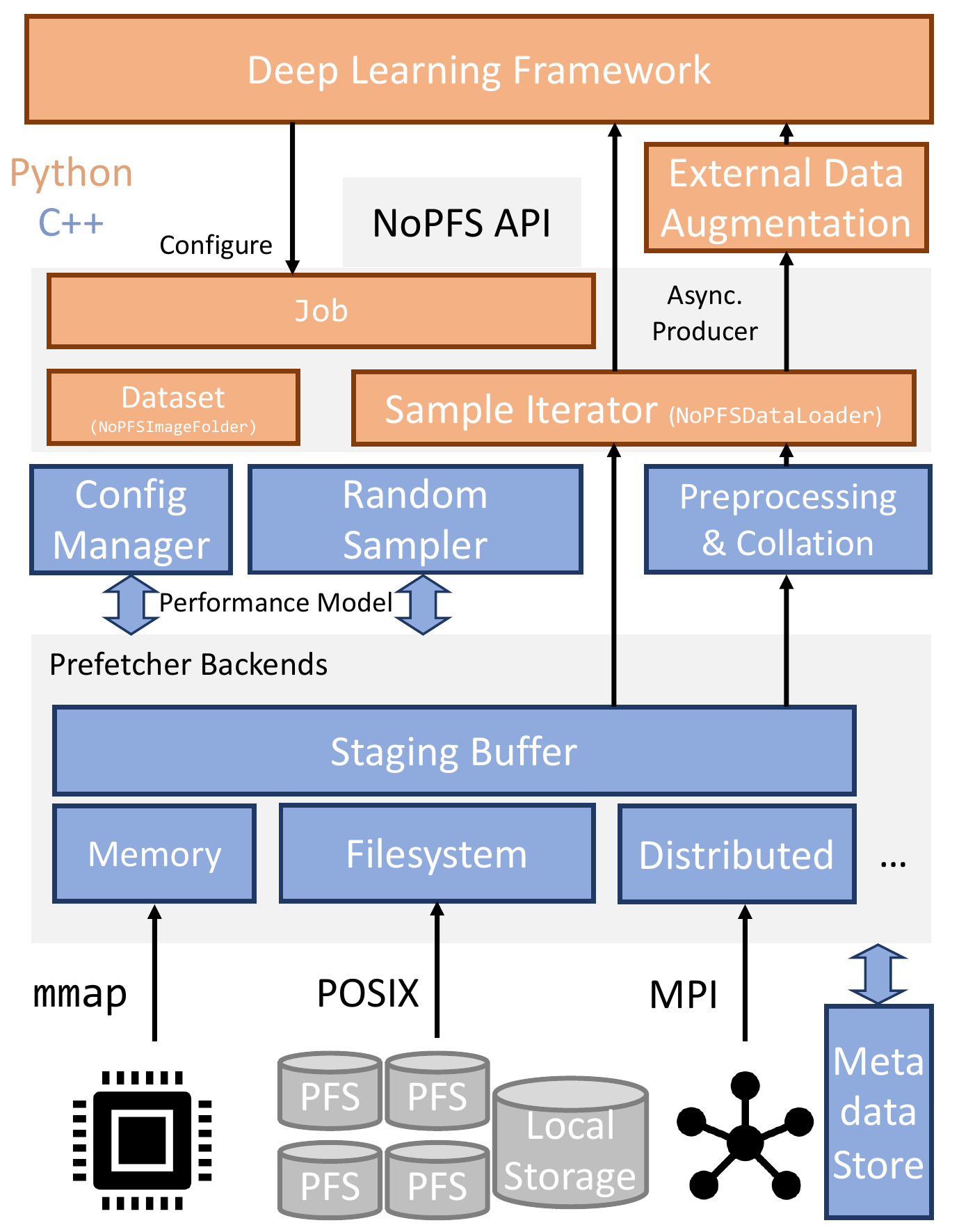}
  \caption{Overview of the NoPFS implementation.}
  \label{fig:nopfs-impl}
\end{figure}

We now briefly describe the implementation of this design, summarized in Fig.~\ref{fig:nopfs-impl}.
NoPFS consists of a core backend implemented in C++ and a generic Python interface that exposes the functionality for integration with existing deep learning frameworks.

\subsubsection{Python Interface}
\label{subsubsec:nopfs-py}

The Python interface provides the \texttt{Job} class, which represents the execution of a machine learning job on a particular dataset.
A single Python process can run multiple jobs at the same time (e.g., training and validation).
This only requires the user to specify a few parameters, such as the dataset, batch size, and the number of epochs.
The random seed that generates the access sequence can either be specified manually by the caller or generated by the library.

Once initialized, the \texttt{Job} exposes two key features: \texttt{buffer\_p}, a pointer to NoPFS's staging buffer, allowing zero-copy access to samples; and a \texttt{get} method, which returns samples and their labels, enabling iterator-style access to data.

It is easy to incorporate this into existing training pipelines.
We provide convenience wrappers to replace the data loader and commonly used datasets in PyTorch.
Using these, minimal changes are required to integrate NoPFS with existing PyTorch codebases, as we demonstrate in Fig.~\ref{fig:py-code}.

\begin{figure}[h]
  \centering
  \scriptsize{PyTorch data loading pipeline:}
  \begin{Python}[emph={data_dir,data_transforms,dataset,num_replicas,n,rank,batch_size,sampler,dsampler,dataloader},
    emph={[2]ImageFolder,DistributedSampler,DataLoader},
    caption={}]
dataset = ImageFolder(data_dir, data_transforms)
dsampler = DistributedSampler(dataset, num_replicas=n, rank=rank)
dataloader = DataLoader(dataset, batch_size, sampler=dsampler)
\end{Python}
  \scriptsize{NoPFS data loading pipeline:}
  \begin{Python}[emph={job,data_dir,batch_size,num_epochs,drop_last,dataset,data_transforms,dataloader},
    emph={[2]Job,NoPFSImageFolder,NoPFSDataLoader},
    caption={}]
job = Job(data_dir, batch_size, num_epochs, 'uniform', drop_last)
dataset = NoPFSImageFolder(data_dir, job, data_transforms)
dataloader = NoPFSDataLoader(dataset)
  \end{Python}
  \caption{PyTorch versus NoPFS data loading.}
  \label{fig:py-code}
\end{figure}

\subsubsection{C++ Core}
\label{subsubsec:nopfs-core}

The core of NoPFS comprises a central manager, generic backends for storage and prefetching, and utilities.
For simplicity, the parameters for our performance model are specified by a system-wide configuration file, with parameterized values (e.g., PFS bandwidth for a given number of readers) inferred using linear regression when the exact value is not available.
This could be generalized to automatically determine performance parameters.

Storage backends need only implement a generic interface, and NoPFS currently supports filesystem- and memory-based storage backends, which are sufficient to support most storage classes (including RAM, SSDs, and HDDs).
Additional backends (e.g., for key-value stores or databases) can easily be added.

For tracking samples, a metadata store keeps a catalog of locally cached samples.
A distributed manager class handles all distributed operations among workers, using MPI for the underlying communication infrastructure.
During setup, it is responsible for distributing a worker's access sequence $R$ to all other workers (an \texttt{allgather}).

It also provides functionality for serving locally cached samples to and requesting samples from remote nodes.
While it is always possible for a worker to know that a sample is not cached by any other worker, it is not possible (without additional metadata traffic) for a worker to know whether a worker that will cache a sample has successfully prefetched it.
As requesting a remote sample that has not yet been cached results in wasted communication and increased stall time, we use a heuristic to estimate when a sample has been cached.
Assuming samples are of comparable size and prefetching is load balanced, if the local prefetching has reached the corresponding access stream location, then the remote worker likely has, too.
Note that the failure of this heuristic is not an error, as NoPFS detects such cases, but we wish to minimize such occurrences due to their performance impact.
We confirmed that, in practice, there are very few false positives.

The core prefetching logic is managed by prefetcher backends, which implement all the logic for prefetching to a particular storage class.
Adding a new prefetcher again only requires implementing a simple, generic interface.
NoPFS provides a memory-based prefetcher and a filesystem-based prefetcher (which uses \texttt{mmap} to access files).
We also implement a special prefetcher for the staging buffer, which is filled in a circular manner.
This prefetcher coordinates with the Python interface via a producer/consumer queue to ensure that the consumer knows when samples are available, and that the prefetcher knows when samples have been consumed (and therefore can be replaced).
If a prefetcher for a local storage class finds that a sample that should be present has not yet been fetched, that prefetcher will retrieve and cache the sample itself, helping to smooth out load imbalance.

Finally, the configuration, storage classes, and prefetchers are managed by a prefetcher manager class, which coordinates the different components.
We also provide convenience utilities, including an optimized, OpenCV-based~\cite{opencv_library} image preprocessing pipeline, and batch collation directly into a pinned memory buffer, which we observed could be a bottleneck otherwise.


\section{Performance Simulator}
\label{sec:perfsim}
\begin{figure*}[ht]
  \centering
  \begin{subfigure}[b]{0.3\linewidth}
    \centering
    \includegraphics[width=\linewidth]{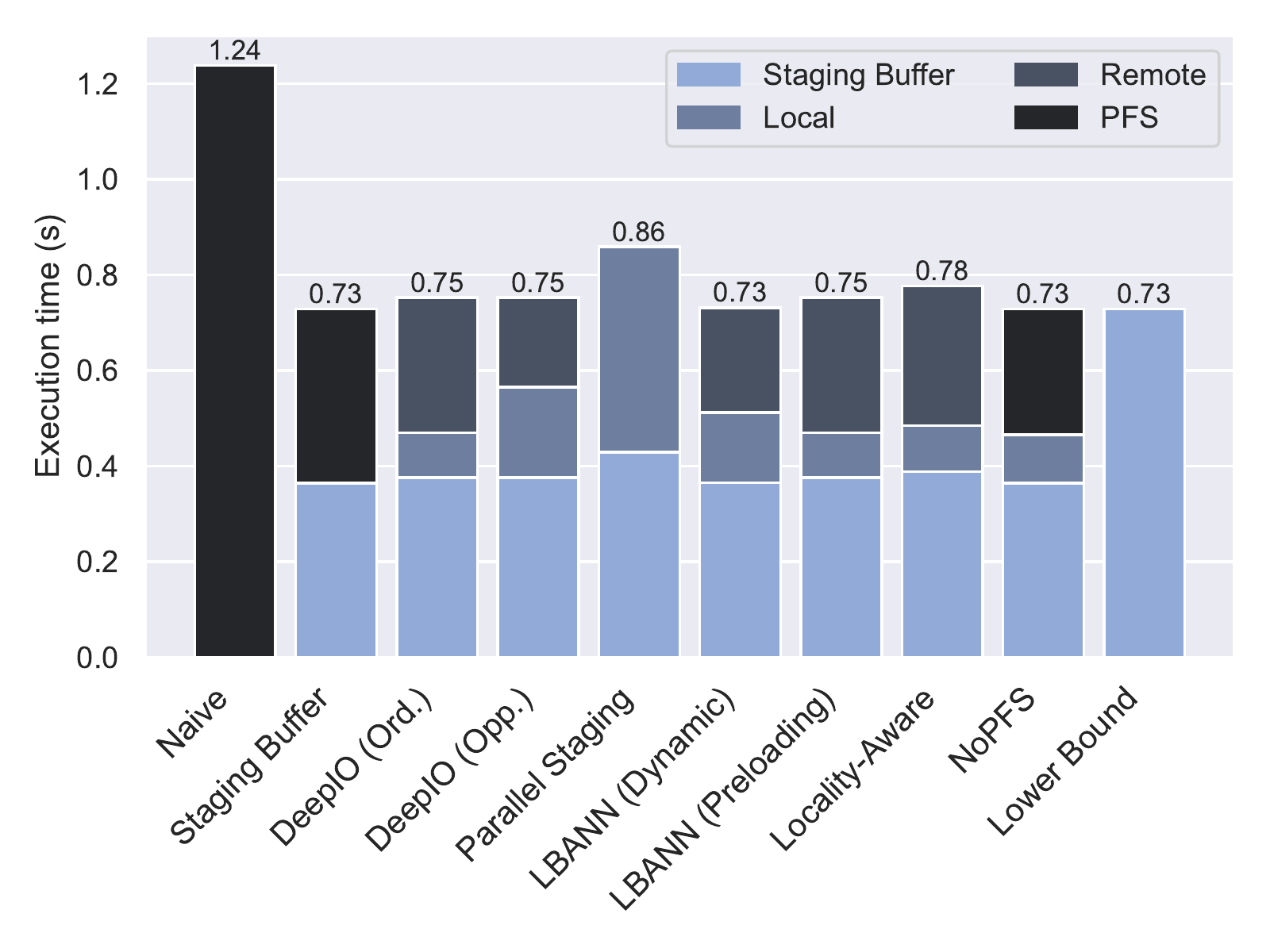}
    \caption{$S < d_1$, MNIST}
    \label{fig:perfsim-mnist}
  \end{subfigure}
  \begin{subfigure}[b]{0.3\linewidth}
    \centering
    \includegraphics[width=\linewidth]{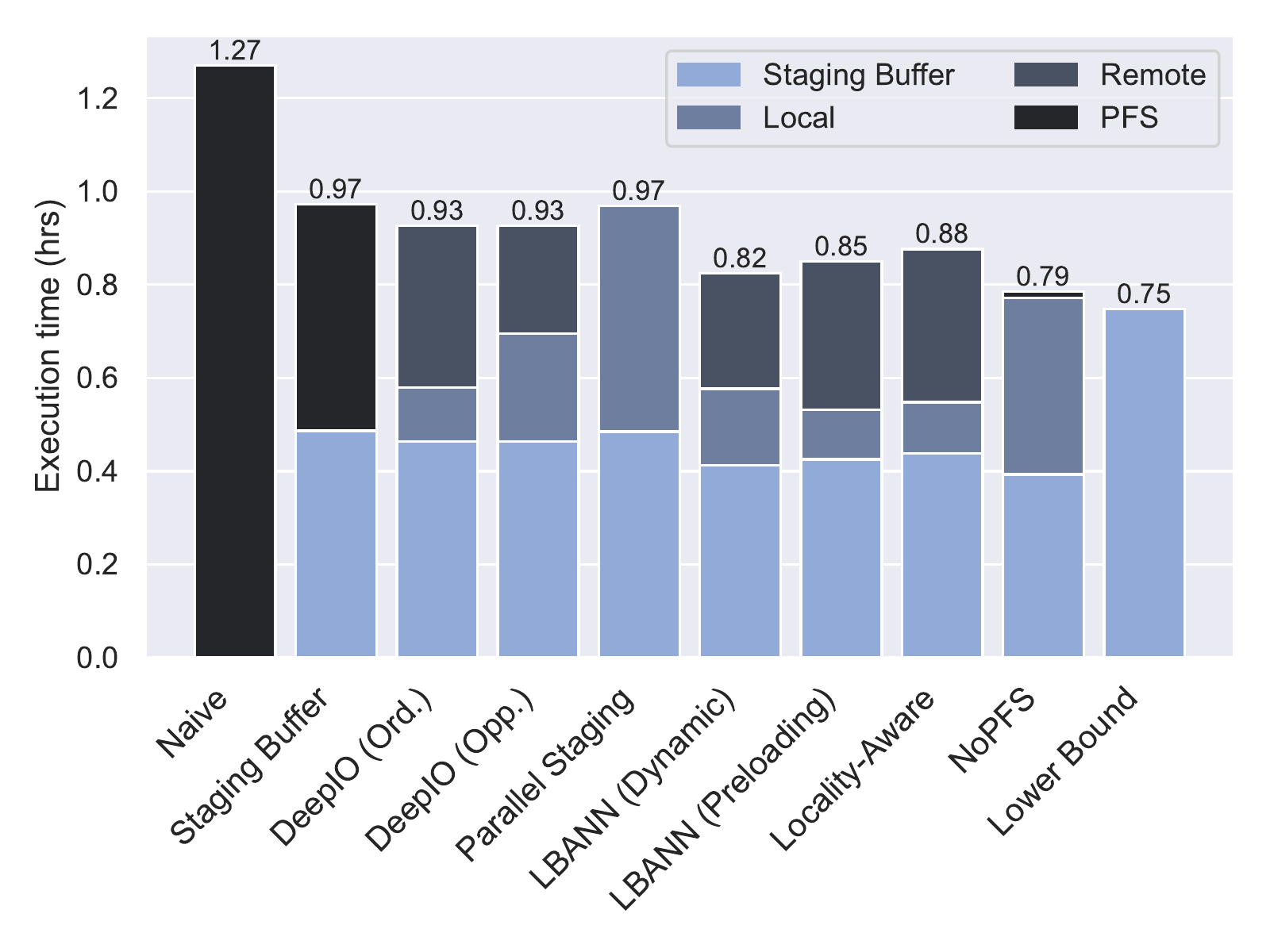}
    \caption{$d_1 < S < D$, ImageNet-1k}
    \label{fig:perfsim-imagenet1k}
  \end{subfigure}
  \begin{subfigure}[b]{0.3\linewidth}
    \centering
    \includegraphics[width=\linewidth]{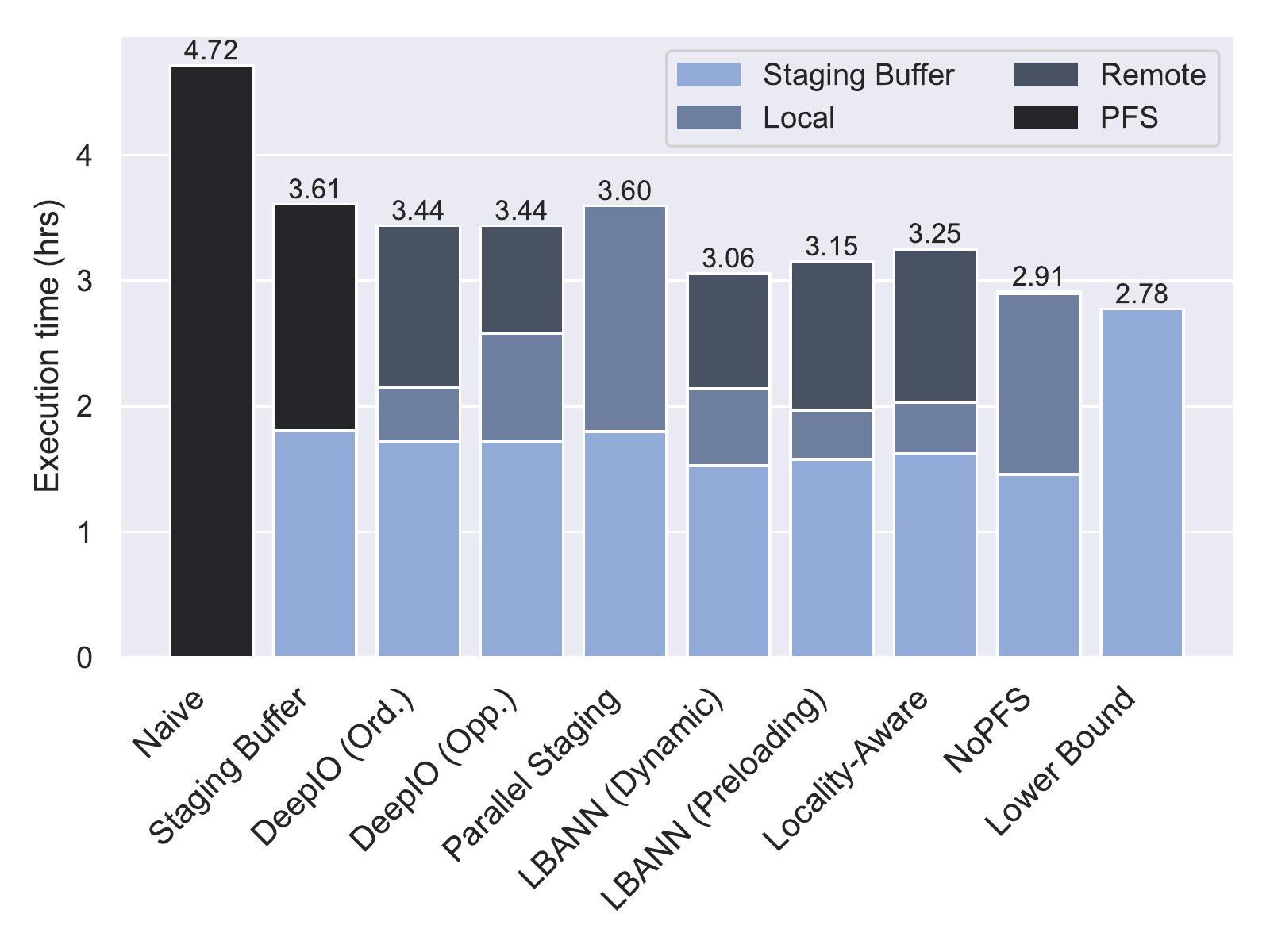}
    \caption{$d_1 < S < ND$, OpenImages}
    \label{fig:perfsim-openimages}
  \end{subfigure}
  \begin{subfigure}[b]{0.3\linewidth}
    \centering
    \includegraphics[width=\linewidth]{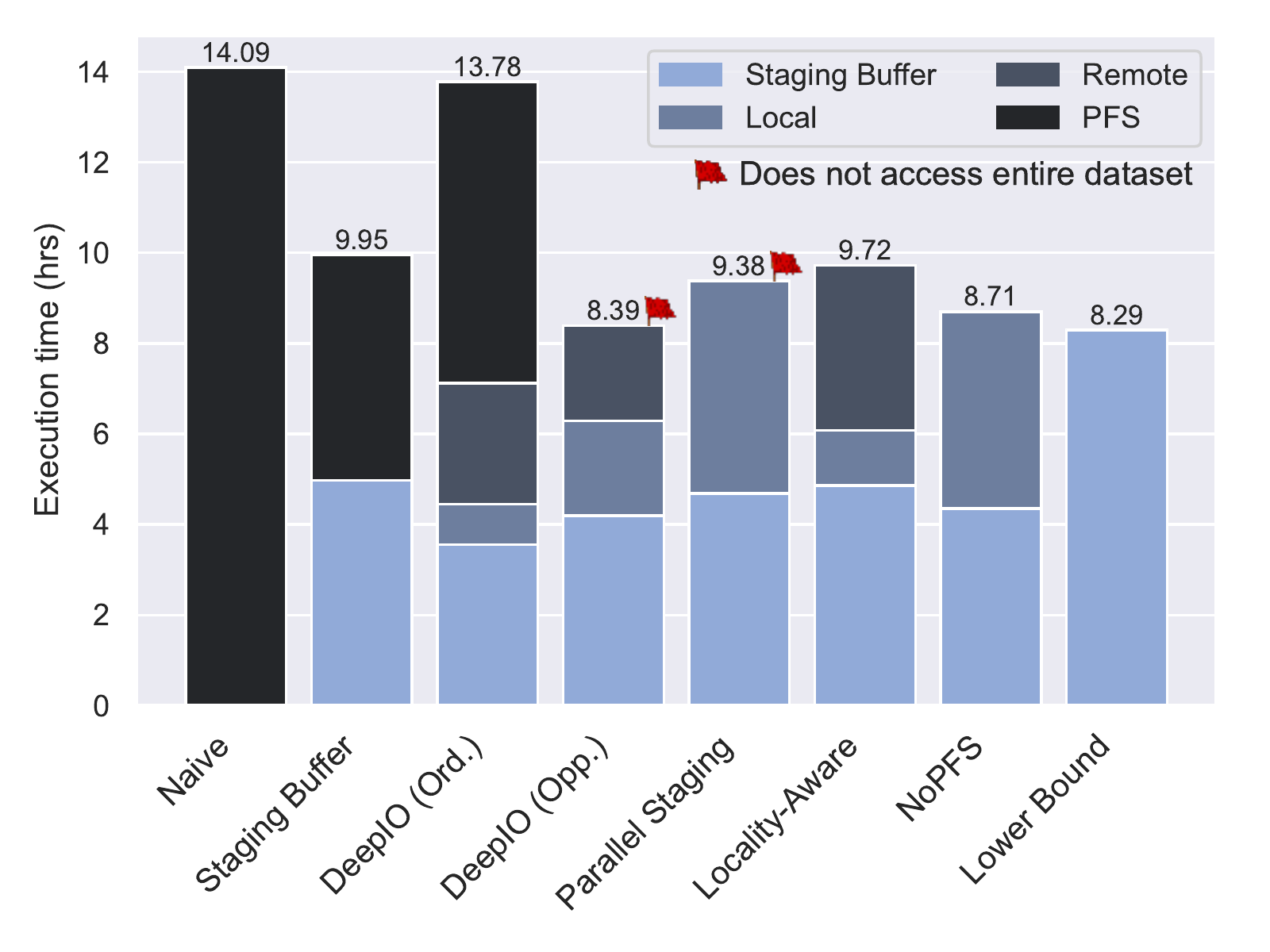}
    \caption{$D < S < ND$, ImageNet-22k}
    \label{fig:perfsim-imagenet22k}
  \end{subfigure}
  \begin{subfigure}[b]{0.3\linewidth}
    \centering
    \includegraphics[width=\linewidth]{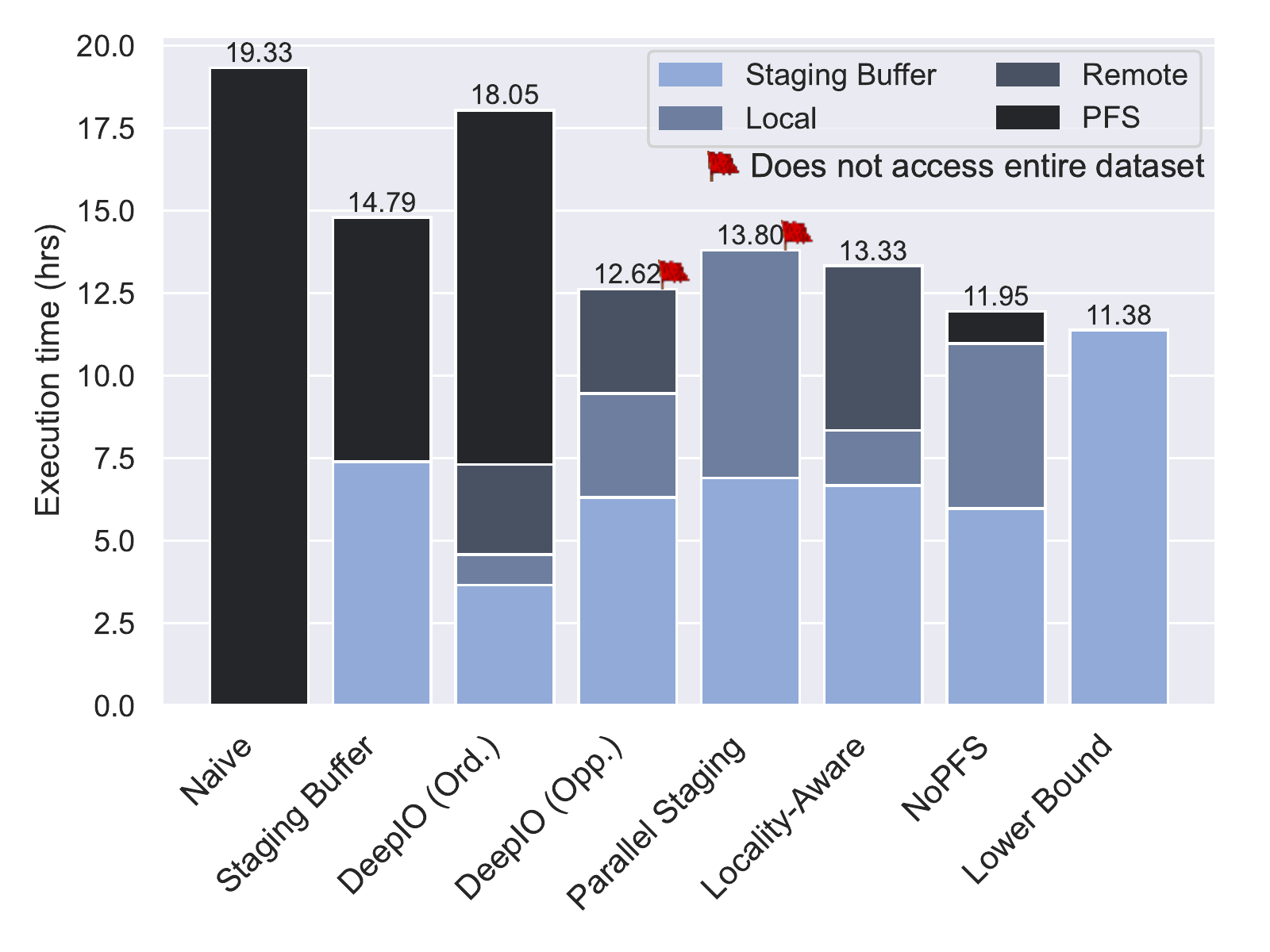}
    \caption{$ND < S$, CosmoFlow}
    \label{fig:perfsim-cosmoflow}
  \end{subfigure}
  \begin{subfigure}[b]{0.3\linewidth}
    \centering
    \includegraphics[width=\linewidth]{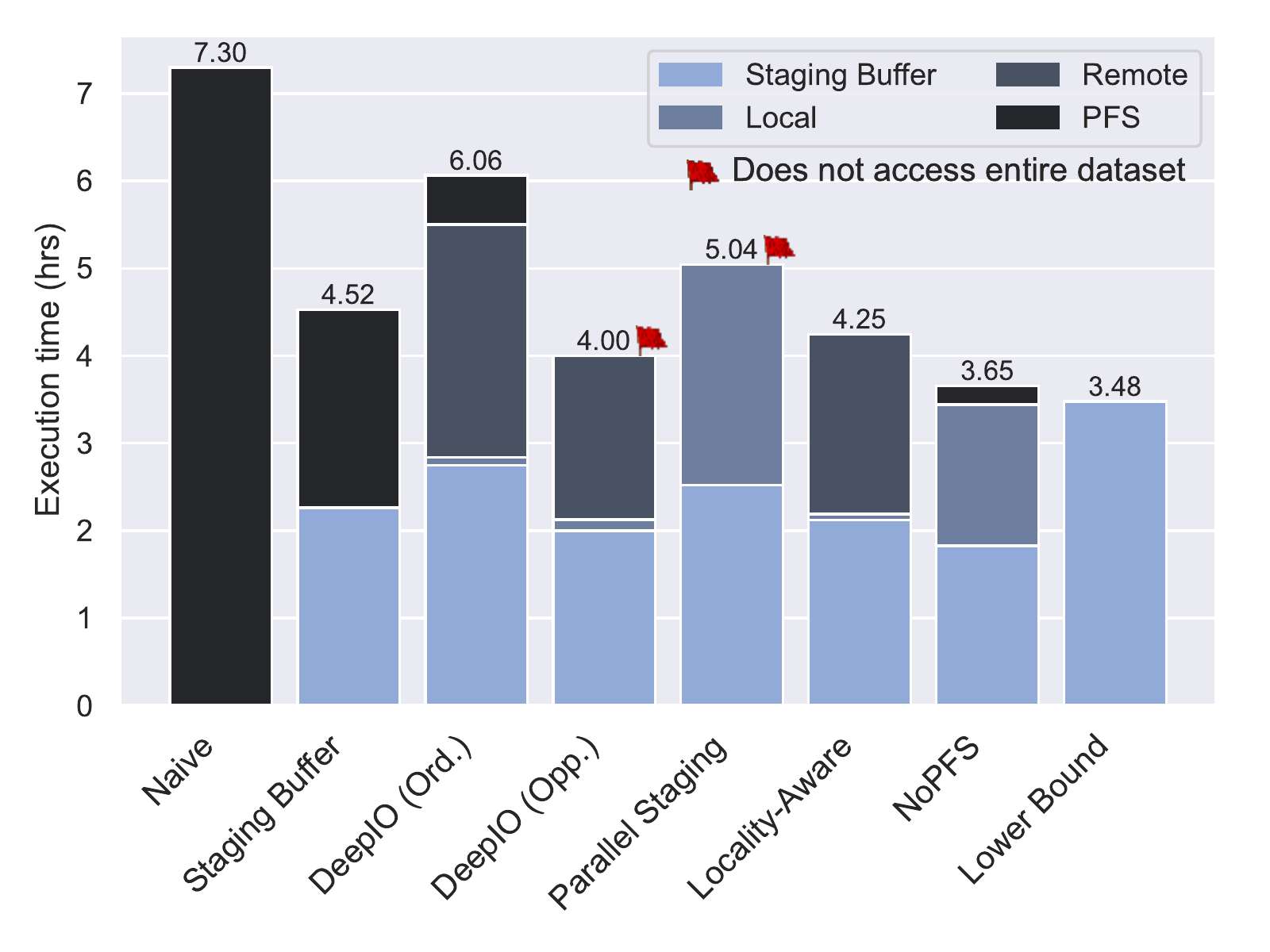}
    \caption{$ND < S$, $N = 8$, CosmoFlow $512^3$}
    \label{fig:perfsim-cosmoflow512}
  \end{subfigure}
  \caption{Performance simulation results. Stacked bars show the proportion of time for each I/O location.}
  \label{fig:perfsim-results}
\end{figure*}

We developed a performance simulator based on our performance model to evaluate different data loading strategies.
The simulator supports arbitrary dataset, system, and I/O strategy configurations.
We do not aim for a precise simulation of training, but rather to capture the relative performance of different I/O strategies.
To this end, we adapt the performance model from Sec.~\ref{sec:perfmodel}, where compute/communication throughput is based on $c$, and assume I/O is overlapped to the greatest extent possible.
We focus on four cases:
\begin{enumerate}
\item $S < d_1$: The dataset fits into the first storage class (typically RAM) of each worker.
  This should not be a challenging situation, but is nevertheless important, as it occurs with small datasets or workers with large amounts of RAM.
\item $d_1 < S < D$: The dataset fits in the aggregate storage of a worker.
  This scenario is interesting, as while a worker can cache the entire dataset, it must use multiple storage classes to do so.
\item $D < S < ND$: The dataset can be cached in the aggregate storage of all workers.
  This requires workers to exploit distributed caching and to minimize the number of PFS accesses.
\item $ND < S$: The dataset is too large to be cached even by the aggregate storage of all workers.
  While this is an uncommon scenario today when using many workers, it is interesting to examine, especially as dataset sizes grow in the future.
  Further, this scenario already occurs when large datasets are used on small training clusters.
\end{enumerate}

We study the following policies:
\begin{enumerate}
\item \texttt{Perfect}: This simulates the case where no stalls occur and provides a lower bound, although it is not realistic in practice.
\item \texttt{Naive}: Loading from the PFS with no prefetching or caching.
\item \texttt{StagingBuffer}: This fills a staging buffer according to the reference string, fetching data from a given location and dropping it after it is consumed.
  When configured to prefetch data from the PFS, this simulates the double buffering or \texttt{tf.data} policies.
\item \texttt{DeepIO}: This simulates the ordered and optimistic modes for DeepIO~\cite{zhu2018entropy}.
  The latter mode may change the access order.
\item \texttt{ParallelStaging}: This simulates data sharding, which also changes the access order, as only locally-available samples are accessed by a worker.
\item \texttt{LBANN}: This simulates the LBANN data store~\cite{jacobs2019parallelizing} (dynamic and preloading approaches).
  As this only caches data in memory, it will fail if the dataset exceeds the aggregate worker memory.
\item \texttt{LocalityAware}: This simulates the locality-aware approach of \citet{yang2019accelerating}.
  When using this policy, we reorder batches at the beginning of the simulation to correspond to the logic described in their paper.
\item \texttt{NoPFS}: NoPFS's policy (Sec.~\ref{sec:nopfs}).
\end{enumerate}

\subsection{Simulation Results}
\label{subsec:perfsim-results}

For brevity, we report simulation results for a single setup simulating a small cluster in Fig.~\ref{fig:perfsim-results}.
Each plot summarizes the execution time, and the stacked bars give the proportion of execution time spent fetching from a particular storage class.
We use $N=4$ workers, each on a dedicated node with a compute throughput of $c=64$ MB/s, a preprocessing rate $\beta=200$ MB/s, and an inter-worker bandwidth $b_c = 24{,}000$ MB/s.
We configured a $5$ GB staging buffer, and two further storage levels representing 120 GB of RAM and a 900 GB local SSD.
We use eight, four, and two prefetcher threads per storage level, and set $r_0(8) = 111$ GB/s, $r_1(4) = 85$ GB/s, and $r_2(2) = 4$ GB/s.
For PFS read throughput, we set $t(1) = 330$ MB/s, $t(2) = 730$ MB/s, $t(4) = 1{,}540$ MB/s, and $t(8) = 2{,}870$ MB/s.
These choices were based on benchmarks of the Lassen supercomputer~\cite{lassen}.

We simulate a set of representative datasets; datasets with different filesizes are assumed to be distributed normally and we vary the $\mu$ and $\sigma$ parameters and the number of samples, $F$, to match.
A per-worker batch size of $B=32$ was used, except for the large CosmoFlow datasets, where $B=16$ and $B=1$, respectively.

\textbf{Scenario 1} ($S < d_1$, $\mu = 0.76$ KB, $\sigma = 0$, $F=50{,}000$, $40$ MB, MNIST~\cite{lecun1998gradient}):
This is representative of small research datasets commonly used in practice.
There is relatively little difference in performance for most policies, and they are close to the lower bound.
The exception is \texttt{Naive}, which is $1.7\times$ slower than the best-performing policy, illustrating the importance of proper I/O handling.

\textbf{Scenario 2} ($d_1 < S < D$, $\mu = 0.1077$ MB, $\sigma = 0.1$, $F = 1{,}281{,}167$, $135$ GB, ImageNet-1k~\cite{ILSVRC15}; and $\mu = 0.2937$ MB, $\sigma = 0.2$, $F = 1{,}743{,}042$, $500$ GB, OpenImages~\cite{kuznetsova2018open}):
Here, NoPFS is the best-performing policy, and is very close to the theoretical lower bound.
There are several key factors behind this performance: NoPFS does not require an initialization phase (in contrast to data staging), reduces PFS reads (whereas \texttt{StagingBuffer} policies always read from the PFS), and utilizes all available resources (in contrast to the LBANN data store, which uses only RAM).

\textbf{Scenario 3} ($ND < S$, $\mu = 0.1077$ MB, $\sigma = 0.2$, $F = 14{,}197{,}122$, $1{,}500$ GB, ImageNet-22k~\cite{deng2009imagenet}):
In this scenario, the dataset exceeds the aggregate storage capacity of each worker.
Further, the LBANN data store no longer supports the dataset, as it is larger than aggregate RAM.
The DeepIO ordered mode performs poorly here, since it fetches uncached samples from the PFS and does not consider access frequency for assigning samples.
NoPFS again performs well and approaches the lower bound.
DeepIO and parallel data staging are also able to perform well, as they never access the PFS, mitigating the impact of the large dataset size.
However, \emph{they no longer access the entire dataset}, significantly impacting potential accuracy during training.

\textbf{Scenario 4} ($ND < S$, $\mu = 17$ MB, $\sigma = 0$, $F = 262{,}144$, $4$~TB, CosmoFlow~\cite{mathuriya2018cosmoflow}; and $N=8$, $\mu = 1{,}000$ MB, $\sigma = 0$, $F = 10{,}000$, $10$ TB, CosmoFlow $512^3$~\cite{oyama2020case}):
We simulate two versions of the CosmoFlow dataset, which are representative of emerging scientific workloads.
The standard CosmoFlow dataset, part of MLPerf-HPC~\cite{mlperf-hpc}, consists of a large number of $128^3$ samples.
That dataset is derived from the CosmoFlow $512^3$ dataset, which consists of a smaller number of large, $512^3$ samples that are sliced to produce the $128^3$ samples.
We use 8 workers for CosmoFlow $512^3$.
In both cases, the dataset size exceeds the storage capacity of the cluster.
Performance is similar to Scenario 3, but at larger scale, indicating NoPFS is able to strong scale well with dataset size and cluster resources while still providing access to the full dataset.

\subsection{Environment Evaluation}
\label{subsec:perfsim-env}

\begin{figure}[t]
  \centering
  \includegraphics[width=0.9\linewidth]{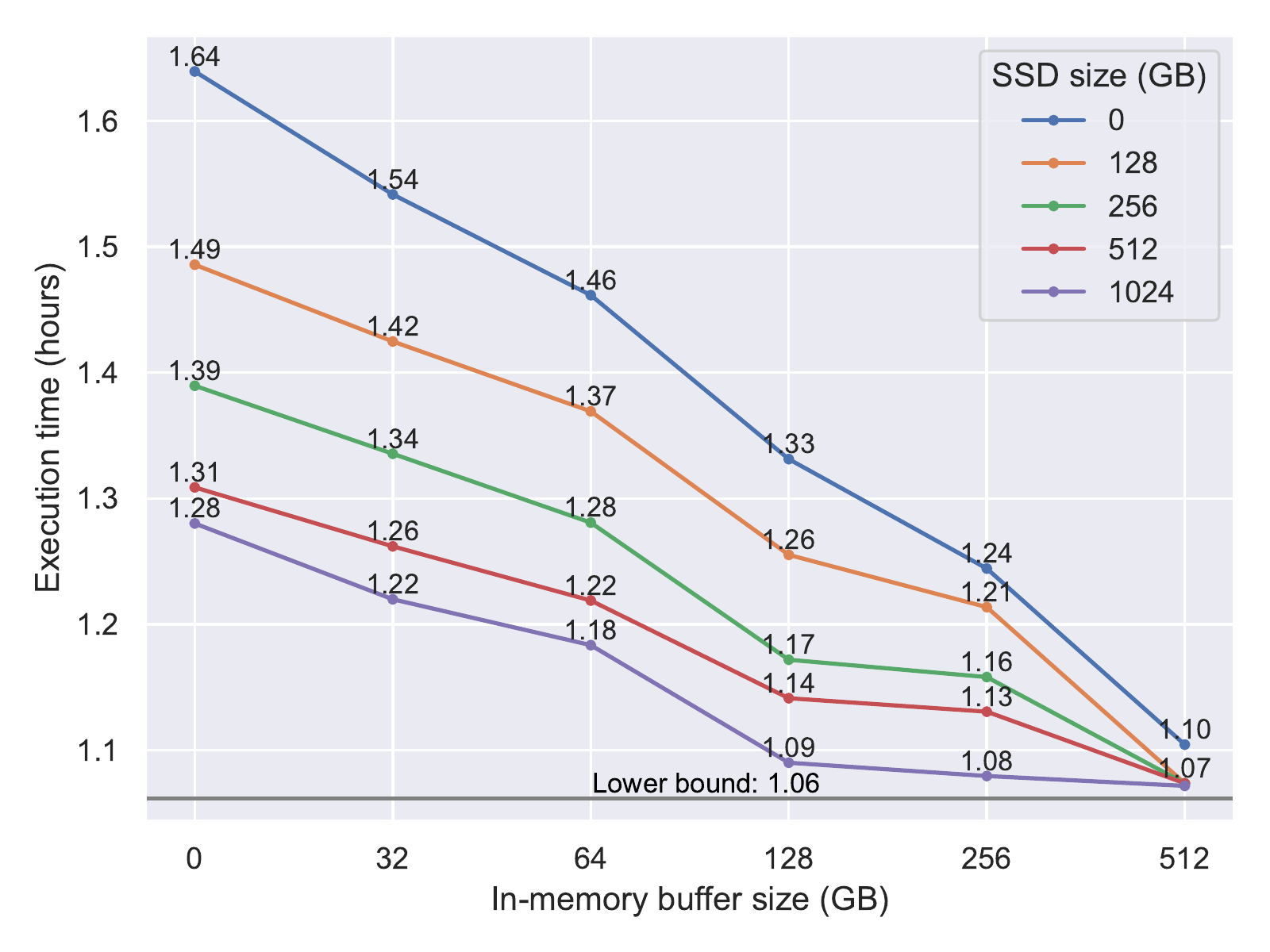}\vspace{-1em}
  \caption{Performance simulation for ImageNet-22k.}
  \label{fig:perfsim-env}
\end{figure}

In addition to comparing I/O policy performance, our simulator can also be used to quantify the impact of changes to a system on training time.
This can be used to identify promising hardware upgrades or when designing new systems to meet training requirements.

To illustrate this, we consider the ImageNet-22k dataset from Scenario 3 with the NoPFS policy and vary the system configuration, assuming $5\times$ compute and preprocessing throughput, which is representative of future machine learning accelerators.
The lower bound on runtime is 1.06 hours.
We first simulate using only a staging buffer of size $1$, $2$, $4$, or $5$ GB, which all resulted in runtimes of 1.64 hours, indicating that the staging buffer is not the limiting factor in this configuration; we fix it at $5$ GB.
We next considered configurations with $32$, $64$, $128$, $256$, or $512$ GB of RAM and $128$, $256$, $512$, or $1024$ GB of SSD as additional storage classes.
The performance for these configurations is illustrated in Fig.~\ref{fig:perfsim-env}.

We observe that, while the best performance is achieved by maximizing total storage, different combinations of storage can often be used to achieve a given performance level if other factors (e.g., cost) need to be optimized for.
Notably, if memory is maximized, then SSD size becomes less relevant.
Alternatively, if memory is expensive, it can be compensated for with additional SSD storage.
This demonstrates why it is critical that an I/O framework be able to automatically adapt to many different storage backends.


\section{Evaluation}
\label{sec:eval}
\begin{figure*}[t]
  \centering
  \includegraphics[width=0.48\linewidth]{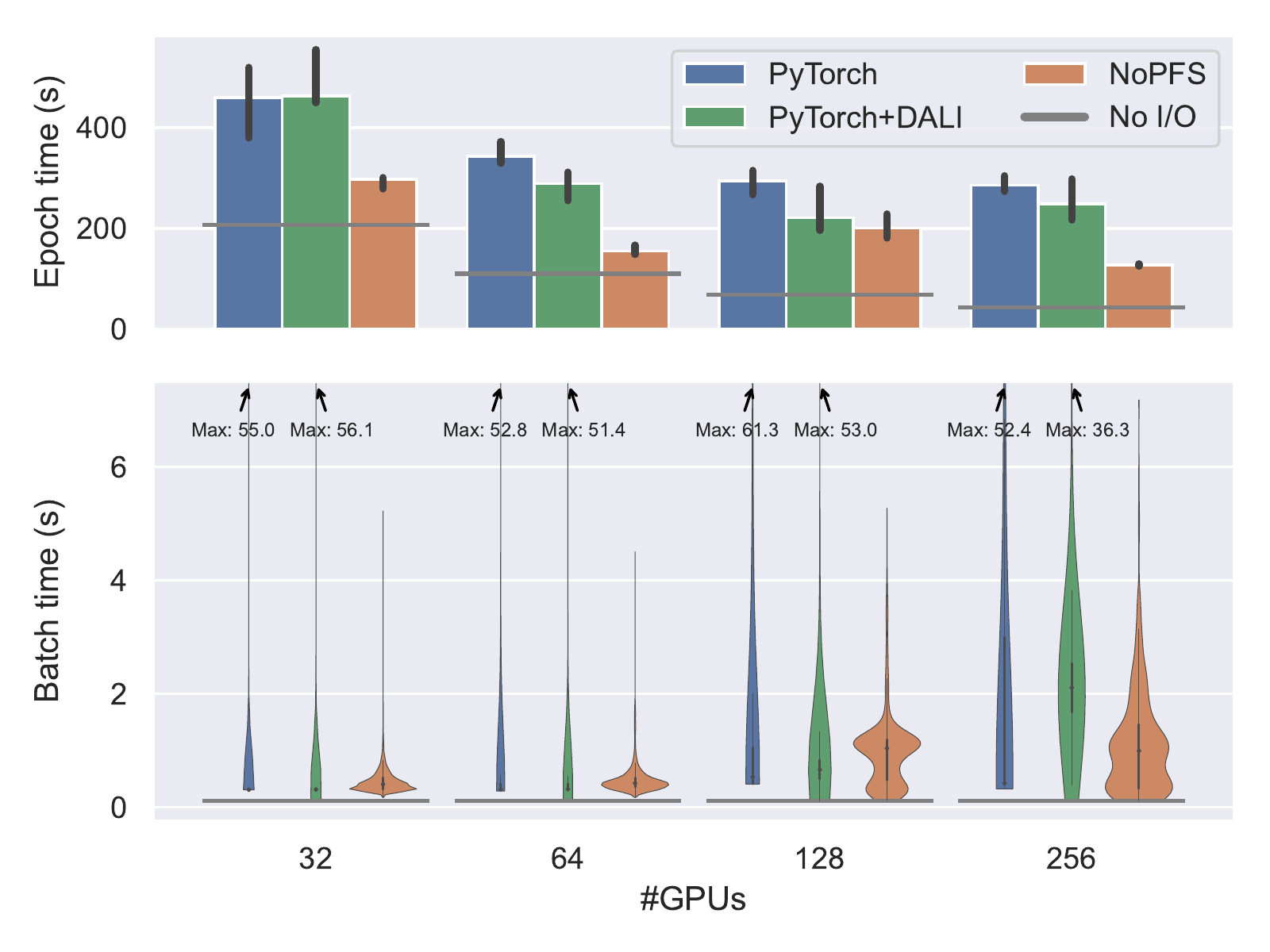}\quad
  \includegraphics[width=0.48\linewidth]{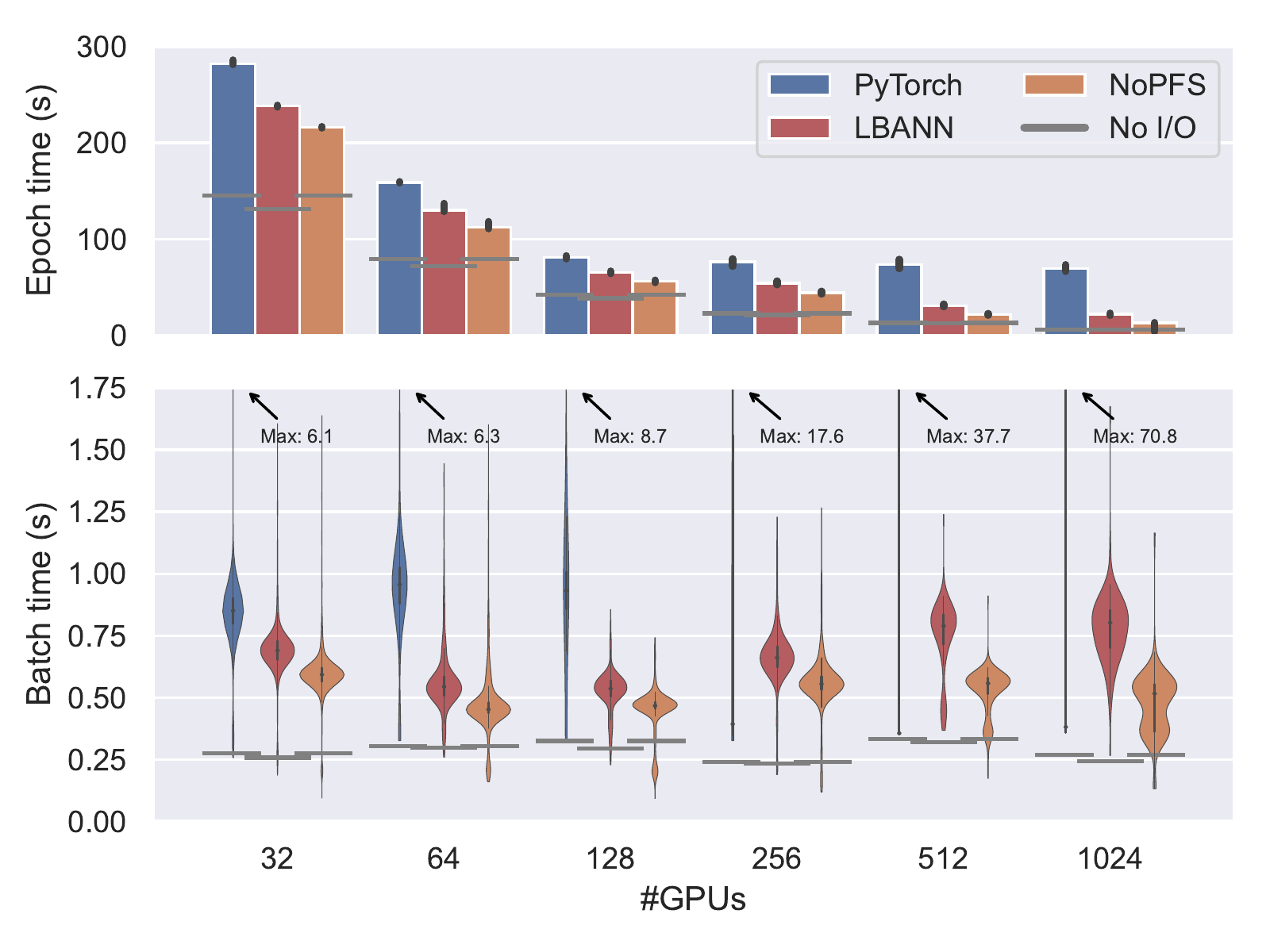}\vspace{-1em}
  \caption{Epoch \& batch time for training ResNet-50 on ImageNet-1k on Piz Daint (left) and Lassen (right) (excl. epoch 0). NoPFS is up to $2.2\times$ faster than PyTorch on Piz Daint and up to $5.4\times$ faster on Lassen; it is also up to $1.7\times$ faster than LBANN.}
  \label{fig:eval-in1k}
\end{figure*}

\begin{figure}[t]
  \centering
  \includegraphics[width=\linewidth]{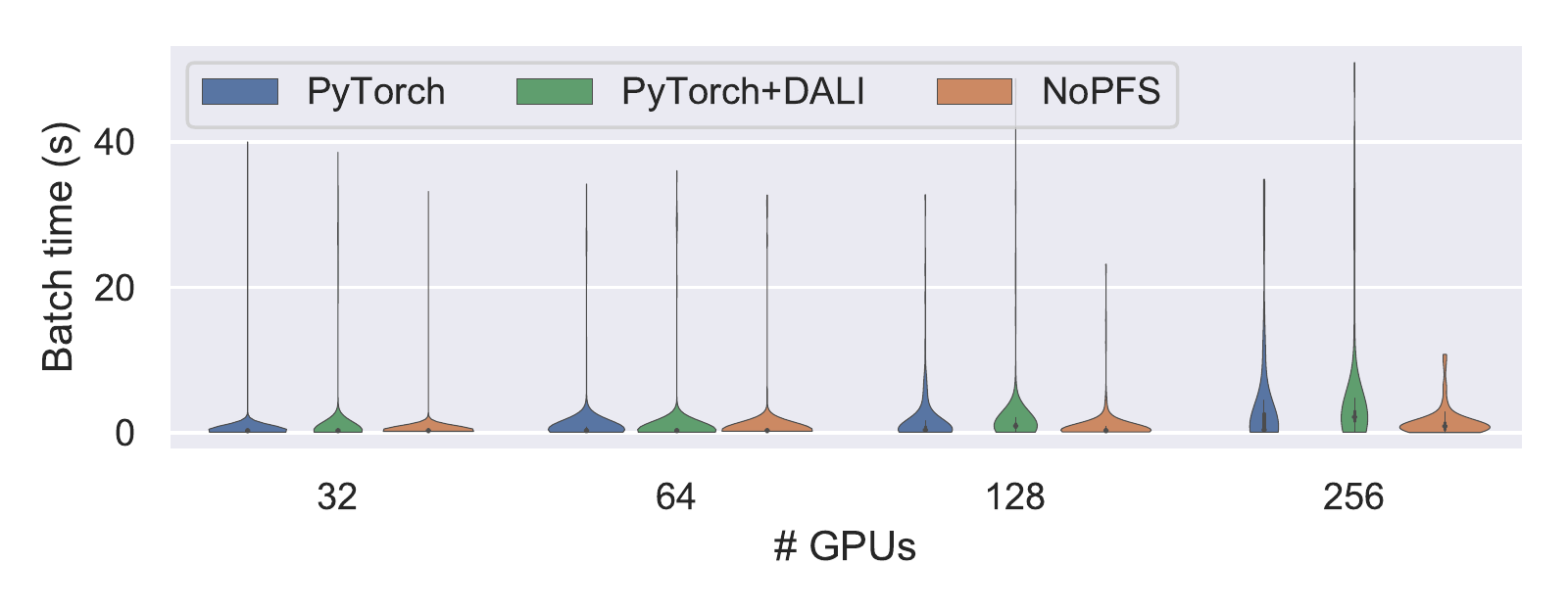}\vspace{-1em}
  \caption{Epoch 0 batch time for ImageNet-1k on Piz Daint.}
  \label{fig:eval-daint-batch-epoch0}
  \vspace{-1em}
\end{figure}

We now experimentally validate NoPFS and compare it to PyTorch using both its built-in \texttt{DataLoader} and DALI~\cite{dali}.
Our experiments use the Piz Daint~\cite{daint} and Lassen~\cite{lassen} supercomputers.
Fig.~\ref{fig:hwinf} provides system details (Lassen uses the same architecture as Sierra).
All runs begin with data at rest on a PFS, in line with MLPerf-HPC guidelines~\cite{mlperf-hpc}.
We perform each run in a separate job allocation to mitigate caching effects.
On Lassen, we use one rank per GPU.

\textbf{Frameworks}~~~~We use PyTorch 1.7 with NCCL2 for all PyTorch benchmarks.
For each model, we endeavored to provide a highly-optimized baseline and all runs use the same training implementation.
We evaluated four different frameworks for I/O:
\begin{itemize}
\item \texttt{PyTorch}: The built-in PyTorch \texttt{DataLoader} and \texttt{Distributed}-\texttt{Sampler}, with multiple prefetching and preprocessing threads.
\item \texttt{DALI}: DALI 0.31.0 for prefetching and preprocessing.
  DALI only supports x86 CPUs, so we only report results for Piz Daint.
\item \texttt{NoPFS}: Our NoPFS implementation, integrated into PyTorch.
\end{itemize}
On Lassen, a NoPFS rank (four per node) uses a 5 GiB staging buffer with eight prefetching threads, 25 GiB of RAM with four prefetching threads, and 300 GiB of SSD with two prefetching threads.
On Piz Daint, it uses a 5 GiB staging buffer with four prefetching threads and 40 GiB of RAM with two prefetching threads.
(We used as much memory as possible without OOM errors.)
To ensure preprocessing is not a bottleneck, we extended PyTorch and NoPFS to perform some data augmentation and conversion on a separate CUDA stream on GPU; DALI does this automatically.
\begin{itemize}
  \item \texttt{LBANN}: LBANN, using its data store in dynamic mode. In this mode, each sample is cached in memory by the worker that reads it first. LBANN requires sufficient memory for its cache.
\end{itemize}

\textbf{Datasets}~~~~We use three datasets, of significantly different sizes and representative of different tasks:
\begin{itemize}
\item ImageNet-1k~\cite{ILSVRC15}: We train ResNet-50 on ImageNet-1k as a standard baseline.
  ImageNet-1k consists of $1{,}281{,}167$ images and $1{,}000$ classes, totalling 135 GB.
  We use the standard data layout, with one directory per class containing all images of that class.
  We use a per-GPU batch size of 64 on Piz Daint and 120 on Lassen.
  Standard data augmentation (random resizes, crops, flips, and normalization) is performed.
\item ImageNet-22k~\cite{deng2009imagenet}: We train ResNet-50 on the larger ImageNet-22k dataset, which consists of $14{,}197{,}103$ samples and $21{,}841$ classes, totalling 1.3 TB.
  This is more representative of larger, emerging datasets used for unsupervised or semisupervised pretraining.
  The configuration is otherwise identical to ImageNet-1k.
\item CosmoFlow~\cite{mathuriya2018cosmoflow}: We use the MLPerf-HPC~\cite{mlperf-hpc} CosmoFlow model and dataset.
  The data consists of $262{,}144$ simulated 3D universes of size $128 \times 128 \times 128$ and four channels, stored in 16-bit integer format, totalling 4 TB.
  Instead of the original HDF5 data format, we converted the data to a simple binary format.
  As HDF5 requires locking to serialize I/O accesses, we found it did not perform well, with median batch times of 3.2 s.
  We use a per-GPU batch size of 16 and log normalization on Lassen.
\end{itemize}

\textbf{Synthetic data lower bound}~~~~To provide a lower bound for training with no I/O and minimal perturbation, we use synthetic data.
We pregenerate random samples in RAM of the appropriate size and data type and use them for training.
The decoding, preprocessing, augmentation, and other aspects of training are otherwise identical to regular training.
We report this as ``No I/O'' in plots.
As this measurement is experimental, some results are slightly faster.
Since LBANN has slightly different performance than PyTorch, its ``No I/O'' performance is measured separately.

\vspace{-0.3em}
\subsection{I/O Performance}
\label{subsec:eval-io-time}

We first compare the training performance of each framework.
We evaluate ImageNet-1k on both Lassen and Piz Daint (Fig.~\ref{fig:eval-in1k}), and the remaining datasets on Lassen (Figs.~\ref{fig:eval-in22k} and~\ref{fig:eval-cf}).
We report median time per epoch with a 95\% confidence interval, using 10 epochs for ImageNet-1k and CosmoFlow and 3 epochs for ImageNet-22k.
We also show violin plots of the per-batch time, skipping the first epoch (which has consistently high variance due to initial data access).
NoPFS consistently has the fastest runtimes and small variance, even on Piz Daint, where the variance of other frameworks is high.

\textbf{ImageNet-1k}~~~~On Piz Daint, NoPFS is $2.2\times$ faster than the PyTorch \texttt{DataLoader} and $1.9\times$ faster than PyTorch+DALI on 256 GPUs.
On Lassen, it is $5.4\times$ faster than PyTorch and $1.7\times$ faster than LBANN on 1024 GPUs.
We observe consistent scaling, except for 128 nodes on Piz Daint, where significant system noise in the NoPFS run resulted in worse performance relative to NoPFS at 64 GPUs.
Despite this, NoPFS is still faster than others.
On Piz Daint, DALI offers a relatively small performance improvement over the default PyTorch \texttt{DataLoader}, likely because its data augmentation pipeline is already well-optimized, including offloading some augmentation to GPU.
Due to PFS contention limiting I/O, PyTorch does not scale beyond 256 GPUs on Lassen.
Despite LBANN being a faster framework than PyTorch in this benchmark (as its no I/O baseline indicates), NoPFS in PyTorch is still able to outperform it.
At small scale, the difference in performance is minimal, but it becomes more significant at larger scale.
This is because LBANN's data store uses a simple first-touch policy for caching samples, and caches each sample in only one location.
Hence, at larger scales, many samples need to be fetched from remote nodes.
While this avoids issues of PFS contention, it is suboptimal compared to NoPFS's access frequency-based caching.

In per-batch runtimes, NoPFS exhibits significantly less variance at all scales than other methods.
Its batches are also fast much more consistently.
This demonstrates a key performance advantage of NoPFS: reducing tail events where read performance is catastrophically slow due to system contention by using local or remote caches.
After the first epoch, PyTorch and DALI exhibit tail events an order of magnitude larger than NoPFS.
We also examined the batch times in the first epoch (Fig.~\ref{fig:eval-daint-batch-epoch0}) on Piz Daint.
NoPFS shows comparable or only slightly lower variance to the other methods, as all must initially access data from the PFS, although NoPFS mitigates this with its prefetching.
However, for PyTorch and DALI, the variance here is comparable to the variance in subsequent epochs: without caching, it is always ``the first epoch'' for a data loader.

\begin{figure}[t]
  \centering
  \includegraphics[width=\linewidth]{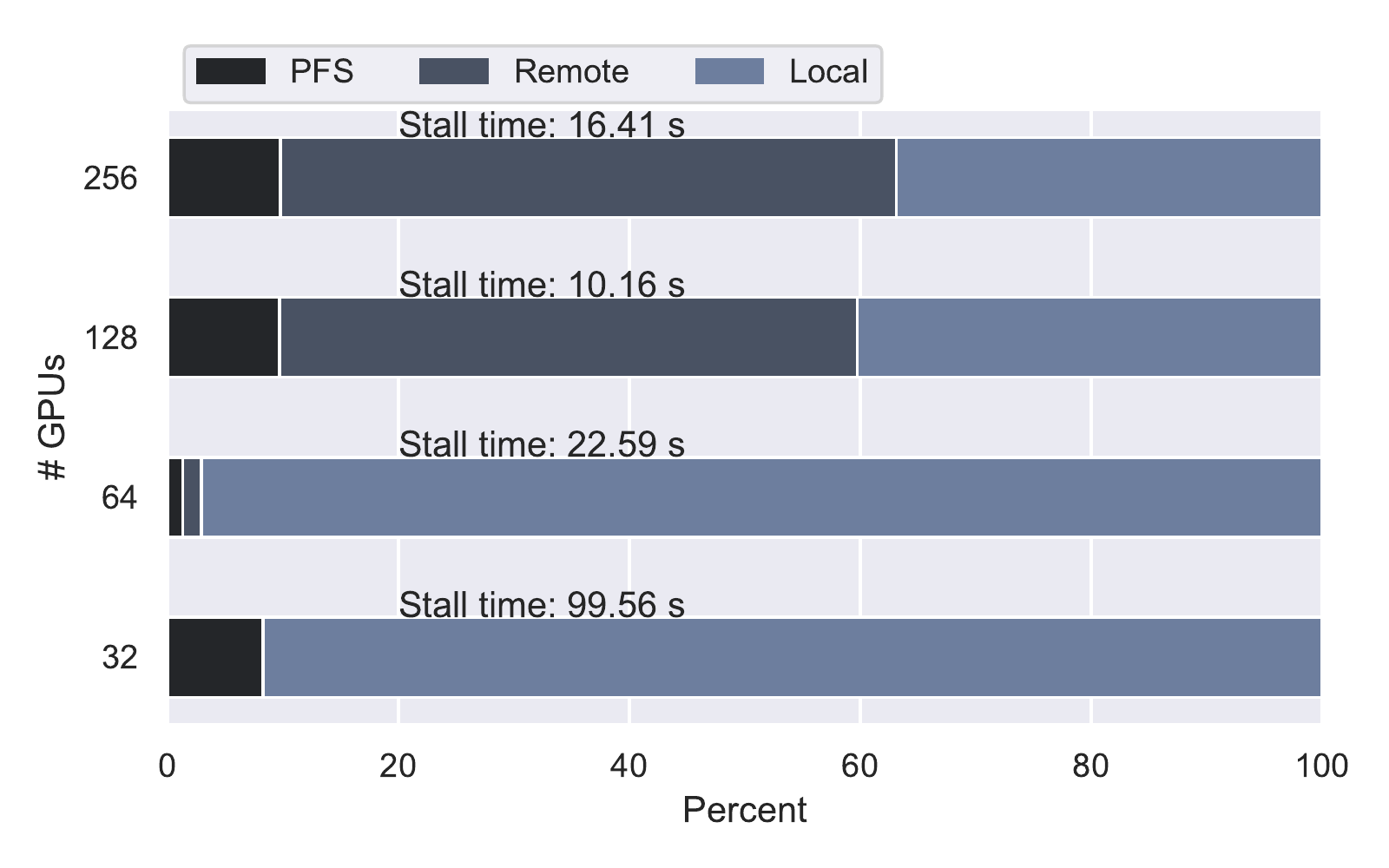}\vspace{-1em}
  \caption{NoPFS cache stats for ImageNet-1k on Piz Daint.}
  \label{fig:eval-daintcache}
\end{figure}

To break down the source of NoPFS's improvements, Fig.~\ref{fig:eval-daintcache} presents the stall time and the percent of staging buffer prefetches that were from local storage, a remote node's cache, or the PFS, aggregated over all epochs.
Stall time decreases at larger scales, as NoPFS is able to strong scale to take advantage of additional storage across the cluster.
The fetch locations also demonstrate how NoPFS adapts to changing cluster conditions.
At smaller scales, the PFS is under less contention, and NoPFS is able to prefetch into on-node memory quickly.
Further, as the number of GPUs increases, each node sees a smaller portion of the dataset, making the prefetching task easier.
However, beyond 64 GPUs, it becomes slower to read from the PFS, and NoPFS instead fetches samples from remote nodes that have already cached them.

\begin{figure}[t]
  \centering
  \includegraphics[width=\linewidth]{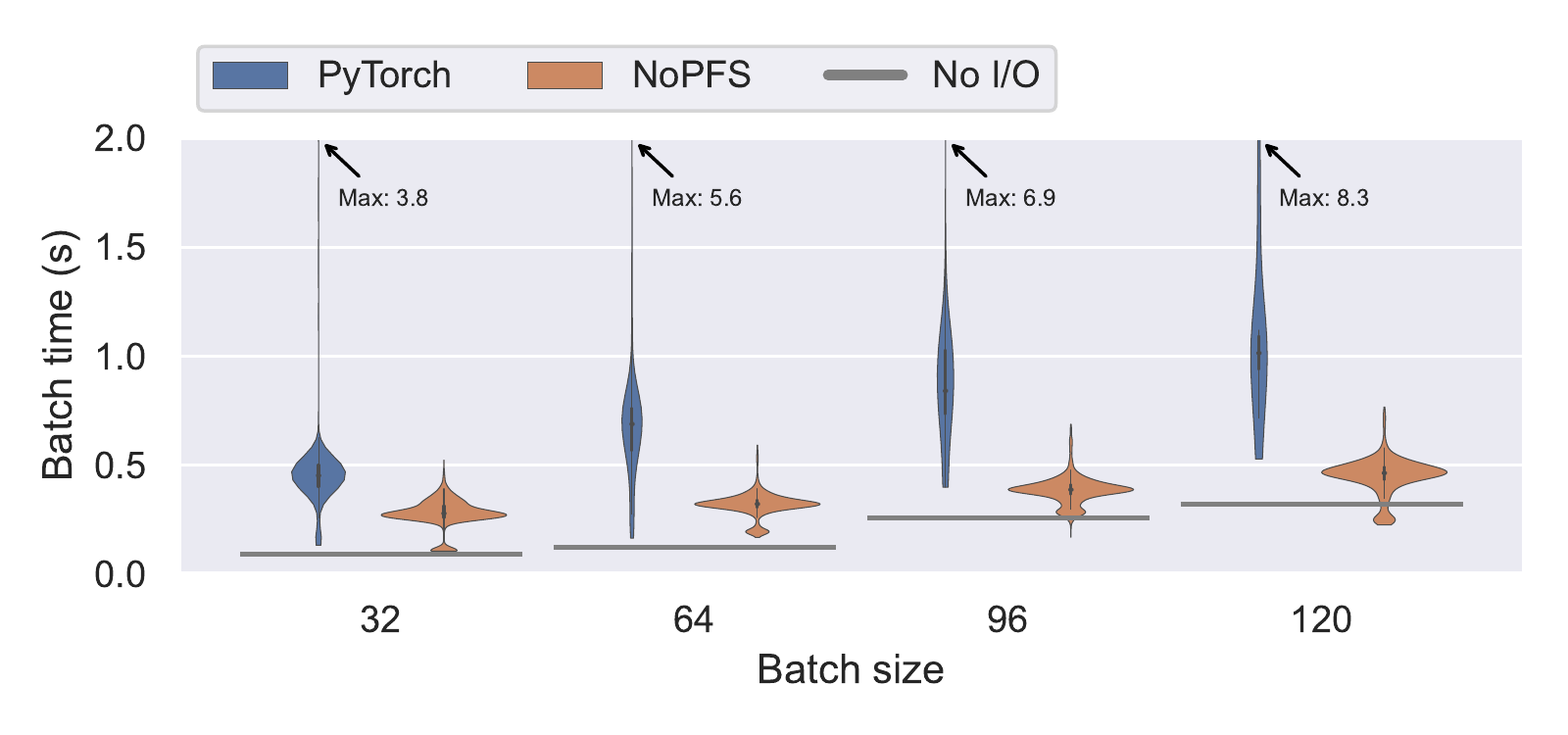}\vspace{-1em}
  \caption{Varying batch sizes for training ResNet-50 on ImageNet-1k on 128 GPUs on Lassen (excl. epoch 0).}
  \label{fig:eval-mb}
\end{figure}

\textbf{Impact of Batch Size}~~~~It is common to vary the batch size when training, depending on how one wishes to trade off memory and learning versus GPU utilization.
To study this effect, we compare NoPFS with PyTorch when training ImageNet-1k on 128 GPUs on Lassen (Fig.~\ref{fig:eval-mb}).
We observe that NoPFS is faster at every batch size (the runtime per batch necessarily increases with larger batch sizes, due to more samples being processed).
Further, while the variance in runtime stays roughly constant for NoPFS, for PyTorch it increases significantly with larger batches, due to additional I/O pressure caused by each rank fetching more data.

\textbf{ImageNet-22k \& CosmoFlow}~~~~Both of these datasets demonstrate similar performance trends as for ImageNet-1k on Lassen.
At 1024 GPUs, NoPFS is $2.4\times$ faster on ImageNet-22k and $2.1\times$ faster on CosmoFlow.
This demonstrates how NoPFS is able to automatically adapt to very different datasets: Either many more samples (ImageNet-22k) or much more data (CosmoFlow).
For CosmoFlow in particular, NoPFS is very close to the no I/O lower bound.
NoPFS also automatically takes advantage of SSDs to cache parts of the CosmoFlow dataset at small scale, when the aggregate node memory is insufficient to hold the dataset.

The batch times for CosmoFlow also exhibit an interesting bimodal distribution.
This is because the samples are all the same, large size (16 MB), leading to significantly different runtimes depending on where the sample is fetched from.

\textbf{Discussion}~~~~While NoPFS shows large performance improvements across systems, scales, and datasets, there is still a gap between its performance and the no I/O lower bound.
We profiled the ImageNet-1k training with 32 GPUs on Lassen and observed that NCCL allreduces took up to $2\times$ longer when performing I/O than without I/O.
We believe this is due to increased contention when performing I/O, as I/O threads interfere with NCCL's communication threads and I/O traffic goes over the same network as allreduces.
This problem is more acute for ImageNet than for CosmoFlow, as the former uses much larger batches and smaller, variable-sized samples, resulting in much more frequent I/O requests.
Indeed, this ``I/O noise'' presents a more general problem: since training is bulk synchronous due to the allreduces in each mini-batch, I/O noise becomes a barrier to scalability~\cite{hoefler2010characterizing,petrini2003case}.
NoPFS helps to significantly reduce this, but better characterizing and mitigating I/O noise (e.g., dedicated I/O cores or storage networks) are important future work.

In general, NoPFS's distributed caching means that samples are read from the PFS as few times as necessary; typically only once for an entire training run when the dataset fits in the aggregate storage.
This has two advantages: NoPFS suffers from less noise due to contention on the PFS, and it has a lower impact on other jobs that may be using the PFS in a shared cluster environment.
Further, as NoPFS scales, it can take advantage of additional distributed memory.
Perhaps counterintuitively, because of very high-speed networks and better random-access performance, reading from \emph{remote} memory can be faster than reading from a local SSD.
While remote accesses increase network traffic, which can interfere with allreduces, not using NoPFS \emph{would still require similar network communication}, since the PFS is accessed over the same network in our systems, while NoPFS avoids PFS contention.
Overall, NoPFS introduces very little overhead (compared to a standard I/O framework, it only needs to compute the access sequence in advance, which is fast) and in practice demonstrates large performance improvements of $2\times$--$5\times$.

\subsection{End-to-end Training}
\label{subsec:eval-e2e}

Finally, we performed end-to-end training of ImageNet-1k using 256 GPUs on Lassen.
We use a batch size of 32 samples per GPU, for a global batch size of 8192, and follow the learning procedure in~\citet{goyal2017accurate}.
The top-1 validation accuracy over time for both NoPFS and PyTorch are shown in Fig.~\ref{fig:eval-e2e}.
In line with our benchmarks, we achieve a $1.42\times$ speedup over the standard PyTorch \texttt{DataLoader} while achieving state-of-the-art accuracy.
Both runs exhibit similar learning curves, albeit with slight variation due to different random seeds for network parameters.
(Note, due to the speedup, NoPFS's curve is compressed.)

\begin{figure}[t]
  \centering
  \includegraphics[width=\linewidth]{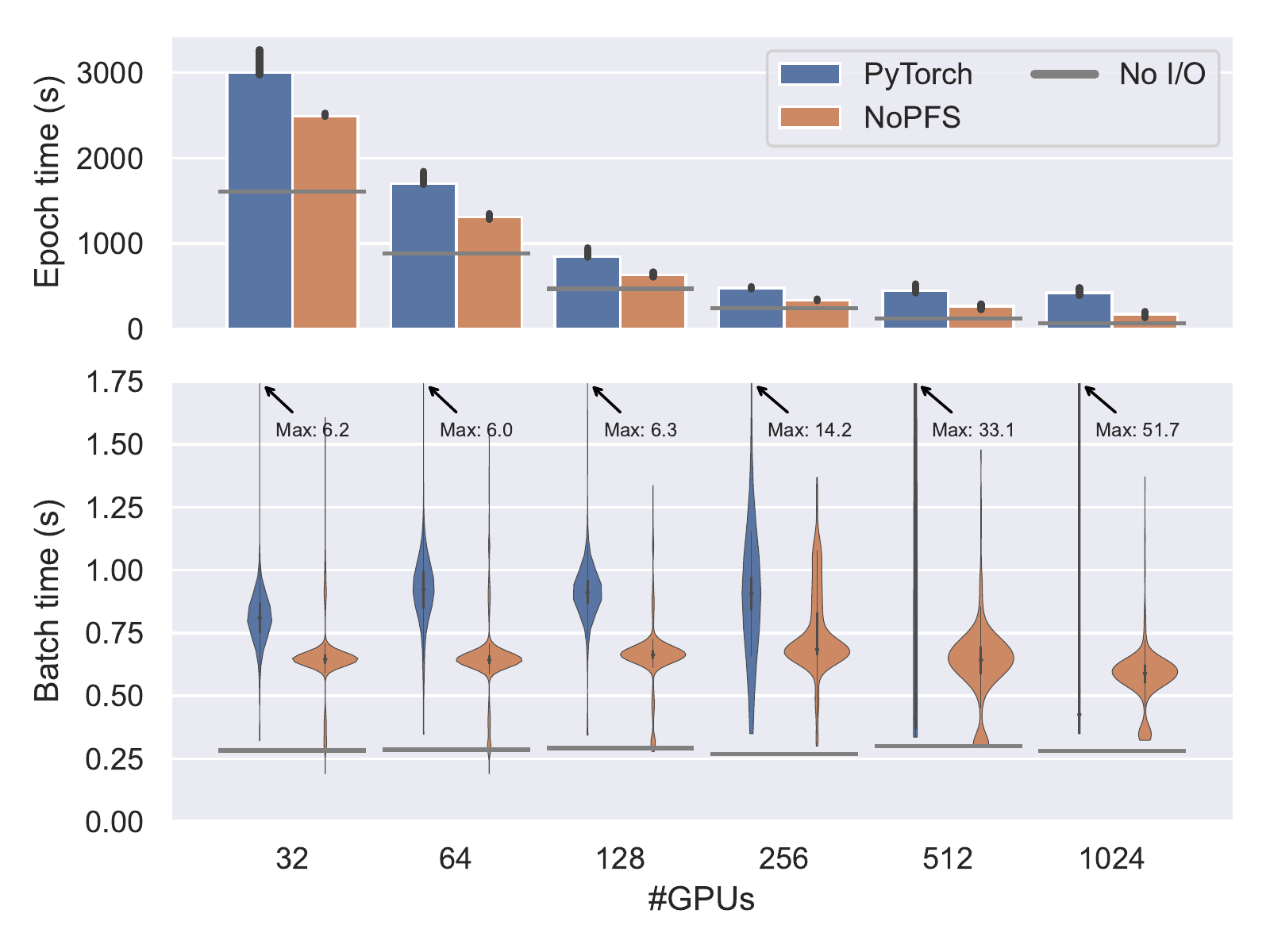}\vspace{-1em}
  \caption{Epoch \& batch time for ImageNet-22k on Lassen. NoPFS is up to $2.4\times$ faster.}
  \label{fig:eval-in22k}
\end{figure}

\begin{figure}[t]
  \centering
  \includegraphics[width=\linewidth]{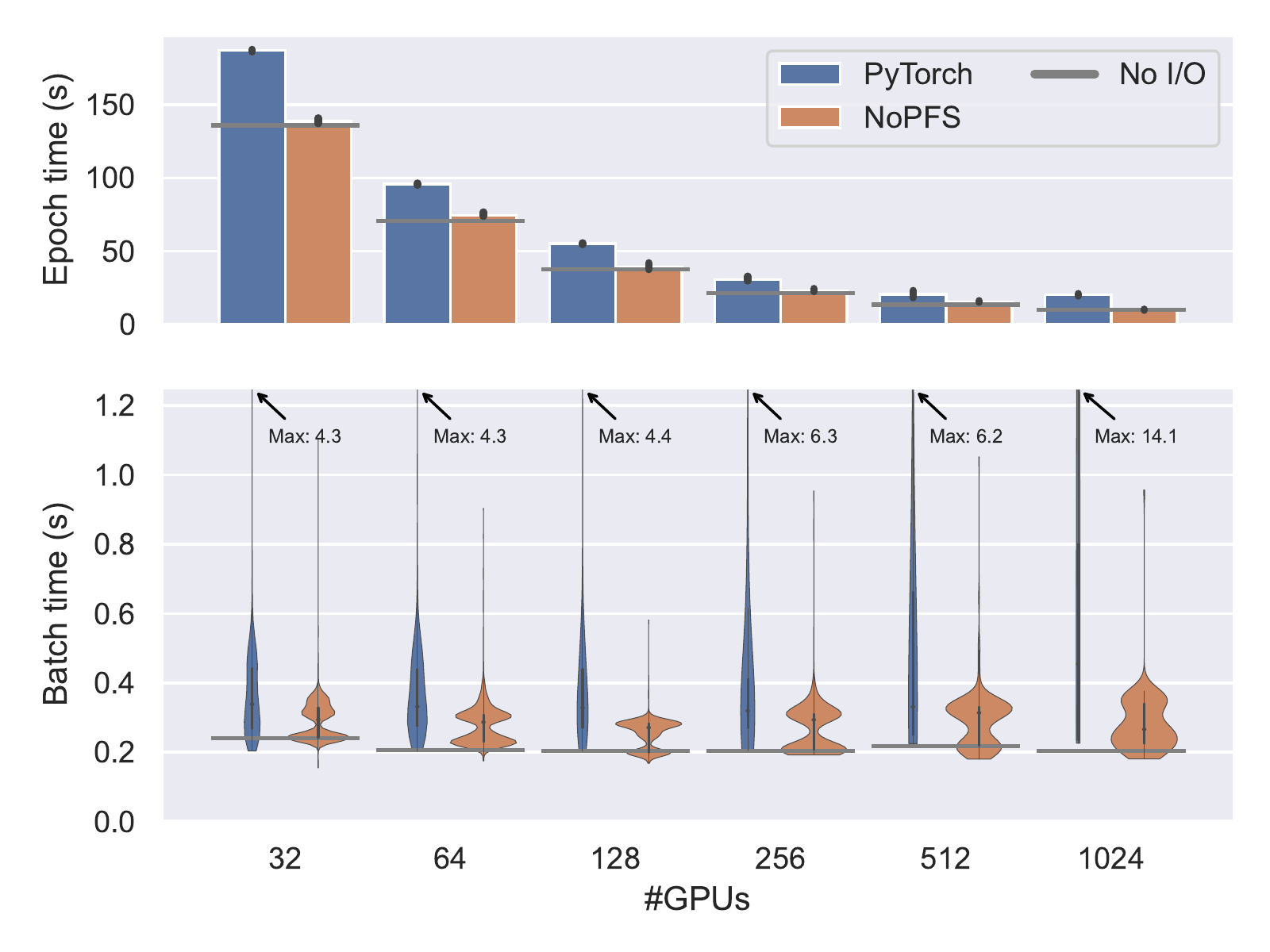}\vspace{-1em}
  \caption{Epoch \& batch time for CosmoFlow on Lassen. NoPFS is up to $2.1\times$ faster.}
  \label{fig:eval-cf}
\end{figure}

\begin{figure}[t]
  \centering
  \includegraphics[width=\linewidth]{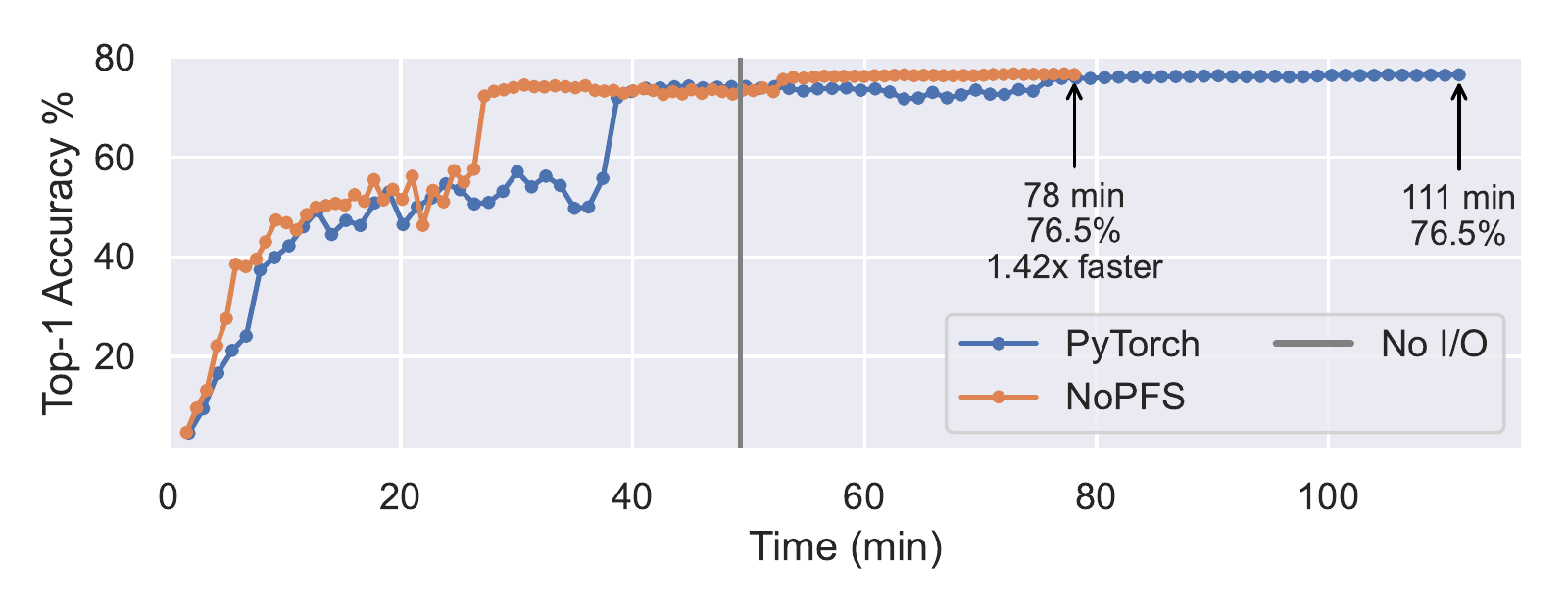}
  \includegraphics[width=\linewidth]{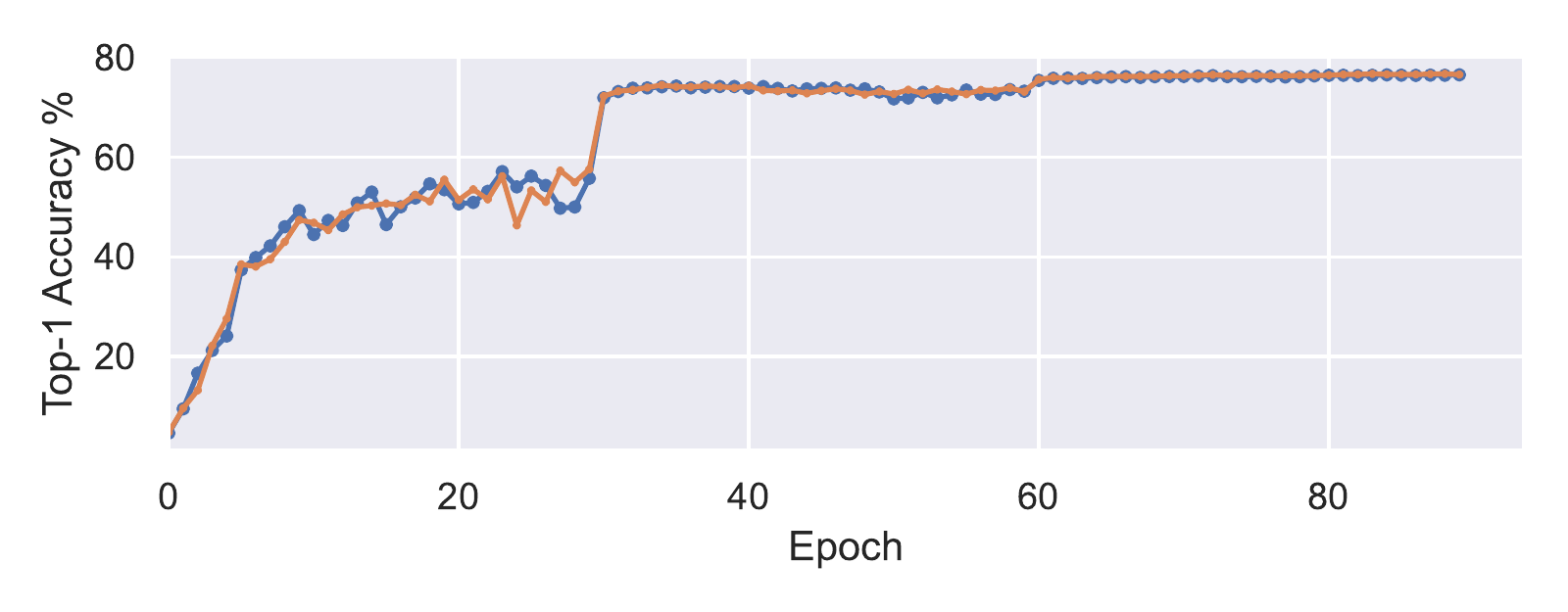}\vspace{-1em}
  \caption{Training ResNet-50 on ImageNet-1k using 256 GPUs on Lassen, accuracy vs time (top) and epochs (bottom).}
  \label{fig:eval-e2e}\vspace{-1em}
\end{figure}


\section{Related Work}
\label{sec:related}
Beyond work on optimizing ML training I/O (see Sec.~\ref{subsec:io-frameworks}), there has been work on optimizing specific aspects or infrastructure.
\citet{pumma2019scalable} optimizes LMDB, Caffe's~\cite{jia2014caffe} I/O subsystem, to address issues in \texttt{mmap} I/O request scheduling.
\citet{chowdhury2019characterization} study the performance of the BeeGFS filesystem for deep learning workloads.
\citet{chien2018characterizing} examine the impact of multi-threading in TensorFlow's I/O pipeline.
Data preprocessing and augmentation can also be a bottleneck during training, as it is typically executed by CPUs, which may be unable to keep up with accelerators.
Optimized pipelines such as DALI~\cite{dali} attempt to address this with careful engineering and by splitting preprocessing between CPU and GPU.

Beyond these, efficient distributed I/O has long been studied in the context of scientific computing applications~\cite{thakur1999data,li2003parallel,howison2010tuning}.
High-performance networks and RDMA have also been used to disaggregate memory and improve I/O performance.
Infiniswap~\cite{gu2017efficient} used RDMA for remote memory paging.
Key/value stores can leverage RDMA~\cite{jose2011memcached} or OpenSHMEM~\cite{fu2017high} for improved performance.
A similar set of work exists for distributed filesystems, which also leverage non-volatile memory, including the Hadoop Distributed Filesystem~\cite{islam2016high}, Octopus~\cite{lu2017octopus}, and Crail~\cite{stuedi2017crail}.
Additionally, distributed and hierarchical caching has been studied in other contexts, such as for video-on-demand content~\cite{koch2018category} and content delivery networks~\cite{borst2010distributed}.


\section{Conclusions}
\label{sec:concl}
Clairvoyance has long been an idea used in theoretical studies of prefetching and caching, but has been difficult to translate to practical applications with complex I/O access patterns.
With machine learning, where the access pattern is random, but predictable given the random seed that generates it, there is now an application that fully benefits from this.
Using clairvoyance, we make a probabilistic analysis of access patterns and show that there is almost always an imbalance in the frequency a worker accesses a particular sample, which we combine with a performance model to drive a hierarchical caching and prefetching policy.
NoPFS provides a simple, powerful interface that can be used in existing training pipelines to improve their I/O performance and reduce overall runtime.

As the compute throughput of accelerators continues to grow faster than that of data movement, the cost of I/O---and the importance of optimizing it---will only increase.
Further, storage hierarchies are only getting deeper and more complicated, necessitating dedicated infrastructure to fully utilize them.
Our work here serves as an initial step toward incorporating more detailed analyses of I/O into runtime frameworks.
Future directions could include NUMA- and topology-awareness for data fetching; dynamic cache management, where samples can migrate between caches; and better characterizing I/O noise.
We expect that by building on clairvoyance and other insights, the I/O bottleneck can be addressed.


\begin{acks}
  The authors thank Marc Snir for discussions inspiring this line of research, and Tim Moon, Quincey Koziol, John Ravi, and Suren Byna for feedback.
  This project received funding from the European Research Council (ERC) under the European Union’s Horizon 2020 program (grant agreement MAELSTROM, No. 955513).
  N.D. is supported by the ETH Postdoctoral Fellowship.
  T.B.N. is supported by the Swiss National Science Foundation (Ambizione Project \#185778).
  Experiments were performed at Livermore Computing and the Swiss National Supercomputing Center.
\end{acks}

\bibliographystyle{ACM-Reference-Format}
\bibliography{references}


\end{document}